\pdfoutput=1
\documentclass[a4paper]{article}

\usepackage{microtype}%if unwanted, comment out or use option "draft"
\usepackage{enumerate}
\usepackage{verbatim}
\usepackage{authblk}
\usepackage{amsmath}
\usepackage{amssymb}
\usepackage{amsthm}
\usepackage{wasysym}
\usepackage{graphicx}
\usepackage{thmtools}
\usepackage{hyperref}
\usepackage{xcolor}
\usepackage{thm-restate}
\usepackage{array}

\newcommand{\etal}{\textit{et al.}\ }

\newtheorem{lemma}{Lemma}
\newtheorem{corollary}[lemma]{Corollary}
\newtheorem{theorem}[lemma]{Theorem}

\newcommand{\rbox}[1]{%
    \colorlet{currentcolor}{.}%
    {\color{red}%
    \fbox{\color{currentcolor}#1}}%
}

\title{Exponential lower bounds for history-based simplex pivot rules on abstract cubes}
\author[1]{Antonis Thomas\thanks{athomas@inf.ethz.ch}}
\affil[1]{Department of Computer Science\\
	Institute of Theoretical Computer Science, ETH Z\"urich \\
	8092 Z\"urich, Switzerland}

\date{}

\begin{document}

\maketitle

\begin{abstract}
The behavior of the \emph{simplex} algorithm is a widely studied subject. 
Specifically, the question of the existence of a polynomial pivot rule for the simplex algorithm
is of major importance. Here, we give exponential lower bounds for three history-based pivot rules.
Those rules decide their next step based on memory of the past steps.
In particular, we study Zadeh's least entered rule, Johnson's least-recently basic rule 
and Cunningham's least-recently considered (or round-robin) rule.
We give \emph{exponential} lower bounds 
on Acyclic Unique Sink Orientations (AUSO) of the abstract cube, for all of these pivot rules.
For Johnson's rule our bound is the first superpolynomial one in any context; for
Zadeh's it is the first one for AUSO. Those two are our main results.
\end{abstract}

\section{Introduction} \label{sec:intro}
The existence of a polynomial time pivot rule for the simplex algorithm is a major open problem in the theory of optimization. 
Most known rules have superpolynomial lower bounds by now. For deterministic rules, in particular, it is the case that many of 
them admit exponential lower bounds. Klee and Minty with their seminal paper \cite{kleeminty}, already in 1972, 
gave an exponential lower bound for Dantzig's original pivot rule. Their construction has been heavily studied ever since 
(for example \cite{ghz98},\cite{bp}) and inspired many later lower bounds. % constructions.

In this paper, we are interested in a family of deterministic pivot rules known as history-based (or having memory). 
For those, superpolynomial lower bounds seemed to be elusive until recently. 
Arguably, the most famous history-based
rule is due to Zadeh. Known as the \emph{least entered} rule, it was described in 1980 with a technical report that was reprinted in 2009 \cite{zadeh}.
This rule keeps a history of how many times each improving variable has been used and, at every step, chooses one that minimizes this history 
(a tie-breaking rule takes care of ties). The least entered rule was specifically designed to attack constructions similar to the Klee-Minty by 
using the improving directions in a balanced way (note that in this regard, it is similar to a random walk). With a letter to Klee in the 80s, Zadeh offered a \$1000 prize
to anyone who can prove polynomial  upper or superpolynomial lower bounds for the least entered rule. 
This prize was claimed in 2011, by Friedmann \cite{Friedmann11}, with a superpolynomial lower bound on actual Linear Programs (LP). No non-trivial upper bounds are known 
for this rule.

Another interesting rule was suggested by Cunningham \cite{C79}, known as the \emph{least-recently considered} rule. It fixes an initial ordering on all
improving directions and then selects one in a round-robin fashion, starting from the last direction selected. The history here is to remember which was the last used improving direction. 
Furthermore, the \emph{least-recently basic} rule, which Cunningham attributes to Johnson, was also first discussed in the same paper \cite{C79}.
That rule selects the improving direction that left the basis least recently (in other words the direction whose opposite was selected least recently).
 For a detailed exposition on those and many other 
history-based pivot rules the interested reader should look at Aoshima \etal \cite{aoshima}. %Let us now motivate our current study.

%In this paper, we provide \emph{exponential lower bounds} for all three aforementioned history-based rules, by 
%means of Unique Sink Orientations (USO). 
We provide \emph{exponential lower bounds}, by means of Acyclic Unique Sink Orientations, for all three aforementioned history-based rules.

\subparagraph*{Unique Sink Orientations} (USO) is an abstract framework that generalizes LP \cite{GS} (and other problems \cite{SW}). It was originally 
described by Stickney and Watson \cite{StiWat} and later revived by Szab\'o and Welzl \cite{SW}. Such abstract frameworks 
have received lots of attention since the discovery of the Random Facet pivot rule: Kalai \cite{kalai92} and, independently, 
Matou\v{s}ek, Sharir and Welzl~\cite{msw96} proved subexponential upper bounds for this rule on LP-type problems. It became evident that 
their analysis made use only of combinatorial properties and, thus, 
it was possible for G\"artner to extend their upper bounds in a much more abstract setting \cite{G95}.

The most well-studied such framework is that of USO (e.g. \cite{matousekcount},\cite{count},\cite{gt15}, also see below).
Intuitively, a USO is an orientation of the hypercube graph such that 
every non-empty face has a unique sink (vertex with only incoming edges). The computational problem 
is to discover the unique global sink %(which exists since the whole cube is a face of itself) 
by performing vertex evaluations (each one reveals the orientation of the edges incident to the vertex).
Commonly, acyclic USO (AUSO) constructions have served as lower bounds for pivot algorithms 
(e.g. \cite{matouseklbrandomfacet}, \cite{ss05}, \cite{ms06}, \cite{hz16})
and our lower bounds are also manifested as AUSO.

\subparagraph{Prior work and open questions.} 
Aoshima \etal \cite{aoshima} explore the possibility that there exist AUSO on which history-based pivot rules take a Hamiltonian path.
They prove, with the help of computers, 
that Zadeh's pivot rule admits such Hamiltonian paths up to dimension 9 at least. On the contrary, they show that Johnson's rule (among others) 
does not admit Hamiltonian paths and, so, they ask if it admits exponential paths on AUSO.
%They ask if it is possible to give exponential lower bound constructions for these history-based rules. We provide 
%such lower bounds for Zadeh's and Johnson's rule; in addition, we simplify and slightly improve the known lower bound for Cunningham's rule.

Recently, Avis and Friedmann \cite{cunning} gave the first exponential lower bound for history-based rules. Namely, they prove an exponential 
lower bound for Cunningham's rule on binary parity games (definitions in \cite{cunning}). 
Their constructions translate immediately to linear programs and also AUSO, 
for which\footnote{The exact translation of binary Parity Games to AUSO is explained in \cite{cunning}. Roughly, in their constructed binary parity game the first player
has $5n'$ assigned vertices. This translates to a cube of dimension $5n'$, where the path that the algorithm will take is 
of length $2^{n'}$. Thus, for $n$-dimensional AUSO, that 
is a lower bound of the form $\Omega(2^{n/5})$.}
the lower bound is $\Omega(2^{n/5})$.
%that are
%realizable as linear programs. Their lower bound is $\Omega(2^{n/5})$ 
%on AUSO\footnote{The exact translation of binary Parity Games to AUSO is explained in \cite{cunning}. Roughly, in their constructed binary parity game the first player
%has $5n'$ assigned vertices. This translates to a cube of dimension $5n'$, where the path that the algorithm will take is $2^{n'}$. Thus, for $n$-dimensional AUSO, that 
%is a lower bound of the form $\Omega(2^{n/5})$.}.
The constructions of \cite{cunning} are very complicated 
and, thus, the authors ask if it is possible to prove exponential lower bounds for this rule, on AUSO, in a simpler manner. %We answer this question to the positive.

%Moreover, they compare their construction to the one by Friedmann \cite{Friedmann11} for Zadeh's rule. The latter gives a
Moreover, they compare their construction to the one for Zadeh's rule \cite{Friedmann11}. The latter gives a
family of non-binary parity games (which correspond to linear programs), where Zadeh's rule takes a subexponential
number of steps, of the form $2^{\Omega(\sqrt{n})}$ (where $n$ is the number of variables of the LP). 
Although binary parity games correspond directly to AUSO, the same is not known for non-binary ones.
%Because those parity games are not binary
%they do not directly correspond to AUSO.
Hence, Avis and Friedmann ask \cite{cunning} if superpolynomial lower bounds for Zadeh's rule exist also on AUSO.
In addition, Friedmann's lower bound \cite{Friedmann11} is based on a
tie-breaking rule which is \emph{artificial} in the sense that it always works in favor of the lower bound designer. It is 
not described in the paper because, as the author writes, it is ``not a natural one''.
% is not given explicitly because, as the author claims, it is ``not a natural one''. Essentially,
%that tie-breaking rule always works in favor of the lower bound designer.
Thus, he raises the question \cite{Friedmann11} of whether it is possible to obtain a lower bound with a \emph{natural} tie-breaking rule.
%,which we answer to the positive. 

Finally, Avis and Friedmann write \cite{cunning}: %if one could show exponential lower bounds for any of the history-based rules discussed in \cite{aoshima}. 
``More generally it is of interest to determine whether all of the history based rules mentioned in \cite{aoshima} have exponential behaviour on AUSO''.
%We answer these question to the positive with exponential lower bounds for Zadeh's and Johnson's rules, which constitute our main results.

\subparagraph{Our results.}  In this paper, we report three lower bounds for the aforementioned rules, on AUSO. 
With Theorem~\ref{thm:johnson}, we give an exponential lower bound for Johnson's rule. 
To the best of our knowledge, this is the \emph{first} superpolynomial lower bound for this algorithm in any context.
%Its construction shares similar ideas with the previous proof; it needs, however a more complicated and interesting analysis.
Moreover, we give an exponential lower bound for Zadeh's rule, with Theorem~\ref{thm:zadeh}.
This has a number of advantages compared to the known construction: 
Firstly, it is exponential, whereas Friedmann's lower bounds \cite{Friedmann11} are subexponential
(also, those do \emph{not} translate to AUSO).
Secondly, our constructions are much simpler to describe. Finally, it is based on a tie-breaking rule 
that is essentially as \emph{simple} as possible: a fixed ordered list. 
These two lower bounds constitute the main results of this paper.

With Theorem~\ref{thm:cunn}, we give an exponential lower bound for Cunningham's rule.
%details of the proof will be omitted (but included in Appendix~\ref{app:cunn}). 
The advantage here is that the construction is significantly simpler (our proof can be fully described in a couple of pages, whereas the construction from 
\cite{cunning} takes over 30); the lower bound also happens to be slightly improved.
Theorem~\ref{thm:cunn} serves as a warm-up to the main results by introducing the techniques and notation we use for our constructions.

Therefore, we answer to the positive all the questions described in the previous paragraph.
%It has the drawback that it is not clear if the AUSO produced by our construction can be realized as an LP. 
%The main advantage over the known 
%proof is that this is a very simple and understandable construction that can be fully described in a couple of pages (the constructions from 
%\cite{cunning} take over 30).

\subparagraph{Our methods.} 
%Moreover, let us note on the novelty of %give some intuition on 
%our techniques. %Since, the pivot rules that we study here are deterministic, if they 
%get caught on a cycle, they would cycle forever. This reveals that it does not make sense to give a cyclic 
%orientation for a lower bound. It does, however, give us the idea to ``simulate'' cycles. Instead of running on a cycle
%the algorithms will take a spiral path. This intuition will hold for both our lower bounds. 
The constructions in this paper are based on the building tools originally presented by Schurr and Szab\'o \cite{ss04}; we do, however, introduce some
novel ideas needed to deal with history-based pivot rules.
Most known inductive lower bound constructions (e.g. \cite{ss04},\cite{ss05},\cite{ms06},\cite{hz16})
embed copies of the previous construction into the next one, in such a way that the algorithm gets trapped in the previous construction twice.
For Zadeh's rule this does not work: it balances the
 directions being used
and it inevitably escapes the second trap (at the next inductive step).
To overcome this, we build a trap that consists of a small number of copies, being connected in a 
careful way which ensures that the algorithm uses the improving directions in a \emph{balanced} fashion: it 
follows the path of the previous construction, up to making additional ``balancing moves'' between different copies.

\subparagraph{Lower bounds on AUSO.}
%And a final comment: 
It is not clear if our constructions can be realized as LPs. However, the abstract setting allows for simpler proofs that are easy to communicate. 
%and hints towards 
%the existence of LP-based exponential lower bounds. 
We, thus, believe %(also supported by the recently improved AUSO-based lower bounds for RE \cite{hz16})
that such constructions are relevant for understanding the behavior of the pivot rules we study
and the ideas could be used for the design of LP-based exponential lower bounds. 

For example, 
the first subexponential lower bounds for Random Facet \cite{matouseklbrandomfacet}
(tight to the upper bound; see also \cite{G02}) and for Random Edge \cite{ms06} (at every step chooses one improving direction at random) were 
both proved by AUSO constructions. 
Indeed, for these two rules, subexponential lower bounds have been later proved on actual LP \cite{fhz11}. 
The most recent lower bound on AUSO was by Hansen and Zwick in 2016  \cite{hz16}, where they improve the subexponential
lower bound for Random Edge.
%Note that the lower bound for Random Edge on AUSO was recently improved by Hansen and Zwick \cite{hz16}.
Note that for this rule non-trivial exponential upper bounds are known 
in the general case \cite{hpz14} and under assumptions \cite{gt16}.
%The latter was very recently improved with a new AUSO construction \cite{hz16}, which demonstrates that 
%RE is strictly slower than RF on abstract cubes. 
%Upper bounds for RE are also known in the general case \cite{hpz14} and under assumptions \cite{gt16}.
%For these two rules, subexponential lower bounds have been proved on 
%actual LP \cite{fhz11}.

%A full version of this paper, including all the missing details, can be found at $[12]$.
%
%\AT{Should this be moved to intro? Plus add stuff from Bernds email?}
%Since, the pivot rules that we study here are deterministic, if they 
%get caught on a cycle, they would cycle forever. This reveals that it does not make sense to give a cyclic 
%orientation for a lower bound. It does, however, give us the idea to ``simulate'' cycles. Instead of running on a cycle
%the algorithms will take a spiral path. This intuition will hold for both our lower bounds.

\vspace{-0.1cm}
\section{Preliminaries}  \label{sec:prel}
Let $[n] = \{1, \ldots, n\}$ and $\pm [n] = \{-n, \ldots, -1, 1, \ldots, n\}$. 
Let $Q^{[n]} = 2^{[n]}$ be the set of vertices of the $n$-dimensional hypercube over coordinates in $[n]$.
Often we write $Q^n$ (the superscript indicates the dimension).
A vertex of the hypercube $v \in Q^n$ is denoted by the set of coordinates it contains. 
Generally, with $C \subseteq [n]$ we denote a set of coordinates.

Consider two vertices $v,u \in Q^n$. With $v \oplus u$ we denote the symmetric difference of the two sets.
%$|v \oplus u|$ is the size of the symmetric difference (equivalent to Hamming distance) of the two vertices.
Now, let $C \subseteq 2^{[n]}$ and $v\in Q^n$. A \emph{face} of the hypercube, $F(C,v)$, is defined as the set of vertices 
that are reached from $v$ over the coordinates defined by any subset of $C$,
i.e. $F(C, v) = \{u\in Q^n | v\oplus u \subseteq  C\}$. The dimension of the face is $|C|$. 
We call edges the faces of dimension 1, e.g. $F(\{j\}, v)$. %, and vertices the faces of dimension 0. 
%The faces of dimension $n-1$ are called facets. 
For $k\leq n$ we call a face of dimension $k$ a $k$-face.
%An $(n-1)$-face of $Q^n$ is referred to as a facet.

Let $\psi$ denote an orientation of the edges of the hypercube $Q^n$. 
Consider two vertices $v,u \in Q^n$ and a coordinate $j \in [n]$. The notation $v \xrightarrow{j} u$ (w.r.t $\psi$) means that $F(\{j\}, v) = \{v,u\}$ and 
that the corresponding edge is oriented from $v$ to $u$ in $\psi$.
Sometimes we write $v \rightarrow u$, when the coordinate is irrelevant. An edge $v \xrightarrow{j} u$ 
is forward if $j \in u$ and otherwise we say it is backward. We use $v\rightsquigarrow w$
to denote (that there is) a directed path from $v$ to $w$.

We can now define the concept of \emph{direction}; the algorithms that we study in this paper
have memory of the directions that have been used so far. A direction is a signed coordinate. 
Let $c \in C$ be a coordinate; two different directions correspond to $c$, $+c$ and $-c$. 
At a vertex $v$ the direction $+c$ corresponds to a forward edge incident to $v$ and $-c$ to a backward edge.
We say that a direction is \emph{available} at vertex $v$ if the corresponding edge is outgoing.
Thus, at each vertex if a coordinate is incoming then none of the directions is available and if a coordinate is outgoing
then exactly one of the directions is available. 
Similarly to above, we write $v \xrightarrow{d} u$, for some direction $d$.
Note that if we have $v \xrightarrow{+c} v'$  (similarly $-c$) 
then at $v'$ neither $+c$ nor $-c$ can be available.
Generally, we denote with $D \subseteq \pm [n]$ a set of directions. Given a set of coordinates $C$, we say
that $D$ is the set of directions that corresponds to $C$ to mean $D = \{-c, +c \mid c \in C\}$.
Often, we use $d$ to denote a direction without specifying its sign.

Then, $\psi$ is a \emph{Unique Sink Orientation} (USO) of $Q^n$ when every non-empty face has a unique sink. USO can be either cyclic 
or acyclic (for the latter we write AUSO). $n$-(A)USO means (A)USO over $Q^n$. For a USO $\psi$,
we define $s_\psi$,
the outmap function, in the spirit of \cite{SW}. 
For every $v \in Q^n$, \mbox{$s_\psi(v) = \{ j \in [n] | v \xrightarrow{j} (v \oplus \{j\} )\}$}, that is the set of coordinates on which $v$ 
has an outgoing edge.
%For every $v \in Q^n$, $ s_\psi(v) = \{ j \in [n] | v \xrightarrow{j} v \oplus \{j\} \}$.
A sink of a face $F(C,v)$ is a vertex $u \in F(C,v)$, such that $s_\psi(u) \cap C = \emptyset$. The whole cube is a face of itself;
thus, there is a unique vertex $v$, the \emph{global sink} with $s(v) = \emptyset$.
In the rest of the paper, we write $s(v)$ to denote the outmap of $v$ (the exact orientation $\psi$ will be clear from the context).
Moreover, let us call a USO \emph{uniform} when it is such that every edge is oriented towards the global sink. Of course, such a USO is acyclic.

The computational problem associated with a USO is to find the global sink. The computational model is 
the \emph{vertex oracle} model. We have access to an oracle such that when we give it a vertex $v$, it replies 
with the outmap $s(v)$ of $v$. This is the standard computational model in USO literature and all the lower and upper
bounds are with respect to it.
%\vspace{-0.5cm}
%\subsubsection{Building tools.} 

We are now ready to state the Product and Reorientation lemmas (due to \cite{ss04}) which are the building tools for the lower bound constructions in this paper.
The following constitutes the intuitive description of the Product lemma which is relevant to us:
Consider an $n$-AUSO $A$ (oriented hypercube graph) and take $2^m$ copies of $A$. For every vertex $v\in A$, take an $m$-AUSO $A_v$.
We call this the \emph{connecting frame} for $v$. 
Each copy of $A$ corresponds to a vertex of the frame. Then, each vertex $v \in A$ is connected according to $A_v$.
%Connect vertices in the copies of $A$ with the corresponding frame. 
The result is an $(n+m)$-AUSO. Formally:

\begin{lemma}[Product \cite{ss04}] \label{lem:product}
Let $C$ be a set of coordinates, $C'\subseteq C$ and $\bar{C'} = C \setminus C'$.
Let $\tilde{s}$ be a USO outmap on $Q^{C'}$. For each vertex $u \in Q^{C'}$ we have a USO outmap $s_u$ on $Q^{\bar{C'}}$. 
Then, the orientation defined
by the outmap $s(v) = \tilde{s}(v \cap C') \cup s_{v \cap C'}(v \cap \bar{C'})$
on $Q^C$ is a USO. Furthermore, if $\tilde{s}$ and all $s_u$ are acyclic so is $s$.
\end{lemma}

%The above lemma can be interpret in different ways. This is the one that is relevant to our constructions in the next sections:
%The interpretation of the above lemma that is relevant to our constructions is the following:

The Reorientation lemma, which follows, can be intuitively explained this way: 
if we have a USO and there is a face, such that all the vertices in this face have the same outmap on the edges external to 
the face, then we can reorient this face according to any other USO.

%In the constructions we give in the later sections $m$ will be constant. Specifically, we give inductive constructions where at every step
%we add $m$ coordinates, for constant $m$, by connecting $2^m$ copies of the previous inductive step with $m$-dimensional connecting frames.

\begin{lemma}[Reorientation. Corollary 6 in \cite{ss04}] \label{lem:reorientation}
Let $C$ be a set of coordinates, $C'\subseteq C$ and $\bar{C'} = C \setminus C'$. Let $s$ be a USO on $Q^C$ and let 
$\mathcal{F} = F(C', u)$, for some $u\in Q^C$, be a face of $Q^C$.  If, for any two vertices $v, w \in \mathcal{F}$,
$s(v) \cap \bar{C'} = s(w) \cap \bar{C'}$ and $\tilde{s}$ is a USO on $Q^{C'}$, then the outmap 
\mbox{$s'(v) = \tilde{s}(v\cap C') \cup (s(v) \cap \bar{C'})$} for $v \in \mathcal{F}$ 
and $s'(v) = s(v)$ otherwise is a USO on $Q^C$.
\end{lemma}

%The above lemma states that if we have a USO and there is a face, such that all the vertices in this face have the same outmap on the edges external to 
%the face, then we can reorient this face according to any other USO. 
%The result is a USO. However, acyclicity is not necessarily maintained after the use 
%of this lemma; hence, we will have to argue about it explicitly. 

\vspace{-0.1cm}
\section{A warm-up: Cunningham's Rule} \label{sec:cunn}
\begin{theorem} \label{thm:cunn}
There exist $n$-AUSO such that Cunningham's rule, with a suitable starting vertex and list, takes a path of length 
at least $2^{n/4}$.
\end{theorem}

In this section we will give a proof for the above theorem and introduce the general approach. % for our lower bound constructions.
%We will sketch a proof for the above theorem.
%We start this section with some general comments and definitions that will apply to all our constructions. Firstly, they are inductive.
The section starts with some general comments and definitions that are relevant to all our constructions. Firstly, they are inductive.
Let $A_{i}$ be the $i$th step of the induction. The base case is $A_0$. 
We call $C_i$ a bundle of coordinates; that is the set of coordinates that was added at the $i$th step of induction (and $C_0$ are the coordinates of the base case).
We also define $C_i^+ = \bigcup_{k=0}^i C_k$.
Then, $D_i$ denotes the set of directions that corresponds to $C_i$ and, similarly, $D_i^+$ denotes the set of 
directions that corresponds to $C_i^+$.

Let $v_0^i$ be the \emph{starting vertex} for $A_i$. Consider that there is a token, which is initially on $v_0^i$, 
and at every step moves according to the direction that the given algorithm chooses. 
The path that the token takes from $v_0^i$ to the unique sink of $A_i$ is denoted with $P_i$; its length is denoted with $|P_i|$. 
In addition, let us denote with $s_i$ the sink of $A_i$.

To construct $A_{i+1}$ from $A_i$, we take $2^m$ copies of $A_i$ and connect their vertices with $m$-dimensional connecting frames, for some constant $m$
(Lemma~\ref{lem:product}).
Afterwards, we perform one reorientation (Lemma~\ref{lem:reorientation}), to install a simple balancing gadget. 
The token starts at the starting vertex $v_0^{i+1}$ and walks on a path $P$ (in $A_{i+1}$) until it reaches a vertex that has all coordinates from $C_i^+$ incoming. This vertex corresponds 
to the sink of some copy of $A_i$. 
If we project the path $P$ to only the directions from $D_i^+$ we get exactly $P_i$.
 In the balancing gadget the token will 
be taken back to the vertex that corresponds to the starting vertex for $A_i$. 
The idea is to prove that if we project the rest of the token's path to the global sink, to only the directions from $D_i^+$, we get again $P_i$.
Thus,
$|P_{i+1}| > 2|P_i|$. Let $T(n)$ denote the length of the corresponding paths on an $n$-AUSO. The recursion we get then is 
$T(n+m) > 2T(n)$. This gives rise to exponential lower bounds of the form $2^{n/m}$. For Cunningham's and Johnson's 
rules the constant is $m=4$ and for Zadeh's $m=6$.

\subparagraph{The rule.}
Firstly, let us formally define Cunningham's least-recently considered rule. 
Consider that the algorithm runs on an $n$-AUSO.
It has an ordered list $L$ that contains all $2n$ directions; let $L[k]$ indicate the $k$th direction on the list. There is a marker $\mu$ of which direction 
was used last: if direction $L[k]$ was used at the last step then $\mu = k$. At the next step the algorithm 
will start checking the directions on the list from $L[\mu+1]$
in a cyclic order (so if it reaches $L[2n]$ it continues from $L[1]$) and it chooses the first available one.
Initially, $\mu = 2n$ so that the first direction that the algorithm checks is $L[1]$. An example of a run for this 
algorithm will be given below.

%As we have already mentioned the construction is inductive. So, let $A_0$ be the base case and $L_0$ 
%be the list for it. %The base case $A_0$ is not very important in this case. Let it be $F_3$ for simplicity.
%$L_0$ is defined similarly to above, i.e. $L_0 = (+1,-2,+3,-1,+4,-3,+2,-4)$.

\subparagraph{The construction and relevant notations.}
Let $A_0$ be the base case and $L_0$  be the list for it. For the proof of Theorem~\ref{thm:cunn}, the base case is not very important
but, for the sake of consistency, let us make it a 4-AUSO over coordinate set $C_0 = \{c_{0}^1,c_{0}^2,c_{0}^3,c_{0}^4\}$.
Also, let $L_0 = (+c_0^1,-c_0^2,+c_0^3,-c_0^1,+c_0^4,-c_0^3,+c_0^2,-c_0^4)$. %, where we write $\pm k $ to mean $\pm c_0^k$. 
%We will now give a short sketch of the lower bound. Full details can be found in Appendix~\ref{app:cunn}.
To construct $A_{i+1}$ from $A_i$ we take $2^4$ copies 
of $A_i$ which we connect with three different frames, in light of Lemma~\ref{lem:product}. 
The frames $F_1$ and $F_2$ are given in Figure~\ref{fig:cunnframes}; the frame 
$F_3$ is given in Figure~\ref{fig:cunn_F3}.
The new set of coordinates will be $C_{i+1} = \{c_{i+1}^1,c_{i+1}^2,c_{i+1}^3,c_{i+1}^4\}$.
How the three frames connect the copies of $A_i$ will be explained below. 

%\AT{Write about balance AUSO and how the adversary argument works nad produces the first part of the construction.}
Let us define some notation in reference to the figures below.
An AUSO (in the figures of this section 4-AUSO) is given as a collection of 2-faces on the first two coordinates. All coordinates are labeled.
%No edge is drawn explicitly unless it is backward. 
Each square represents a face
on coordinates $C_i^+$. All of these faces, except $\fbox{B}$, are internally oriented according to $A_i$ (correspond to copies of $A_i$). 
The numbers are indicating 
in which order the token will visit them. We refer to these faces in the text; for example, we write $\fbox{1}$ (in reference to the corresponding figure; here, in reference 
to Figure~\ref{fig:cunnframes}) to mean the face $F(C_i^+, \{c_{i+1}^2\})$.
Given a vertex $v$, we write $v \bot \fbox{1}$ to mean the vertex $v' \in \fbox{1}$ , such that $v \cap C_i^+ = v' \cap C_i^+$
(the corresponding vertex in $\fbox{1}$).
%Consider a vertex $v$ and a face $F(C,u)$. % (described by coordinates $C$) of the hypercube. 
%Let $\bar{C}$ be the complement of $C$, i.e. the set of coordinates external to $F$.  % FORMAL DEF of the above.
%With $v \bot F$ we denote the vertex $v' = (v \cap C) \cup (u \cap \bar{C})$. 
Moreover, we write $\fbox{1} \rightsquigarrow \fbox{5}$ to mean a path from a vertex in \fbox{1}  to the corresponding vertex in  \fbox{5}, using only 
directions from $D_{i+1}$.  In this case, the exact vertex will be clear from the context. 
%Those faces are always on coordinates $C_i^+$. 
The face $\fbox{B}$ is the one that contains the balancing gadget, which is installed
by use of Lemma~\ref{lem:reorientation}. In this construction and the one of Section~\ref{sec:johnson}, 
$\fbox{H}$ is a hypersink (has all edges external to the face incoming). In the construction of Section~\ref{sec:zadeh} there is no hypersink.
The notation of the figures that we just described will also be used in the next sections.

\begin{figure}[htbp] 
  \centering\includegraphics[width = \textwidth]{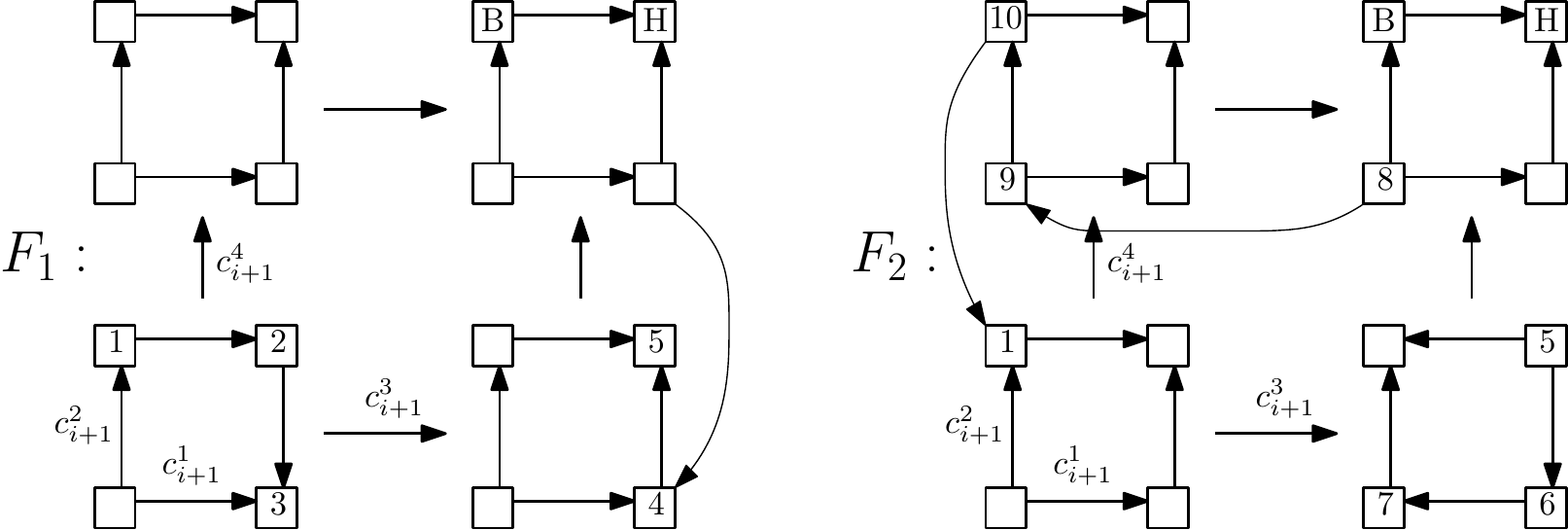}
\caption{The orientations $F_1$ and $F_2$, used as connecting frames, are given in this figure. The 4-dimensional frames are split in 2-faces of coordinates $c_{i+1}^1$ 
and $c_{i+1}^2$. 
In our depictions it is always the case that vertex $\emptyset$ is at the bottom-left whereas the full vertex that contains all the coordinates is at the top-right (\fbox{H} here).
The arrows on the other coordinates indicate the orientation of all the edges on this coordinate except when noted differently. For example, in $F_2$ all
edges on coordinate $c_{i+1}^3$ are oriented from left to right (forward) except the edge $\fbox{9} \leftarrow \fbox{8}$. This notation is valid also for the next figures.
}
\label{fig:cunnframes}
\end{figure} 

Based on Figure~\ref{fig:cunn_F3} below, we will give an example run of the algorithm. In that figure we have 
the connecting frame $F_3$; this 4-AUSO also serves as the base case $A_0$. Let the starting vertex be at $\fbox{1}$, i.e. $v_0^0 = \{c_0^2\}$.
As a reminder, we restate that $L_0 = (+c_0^1,-c_0^2,+c_0^3,-c_0^1,+c_0^4,-c_0^3,+c_0^2,-c_0^4)$.
Then, the algorithm will use $+c_0^1,+c_0^3$ and go to $\fbox{5}$ (notice that $-c_0^2$ was not available). From $\fbox{5}$ it will use $-c_0^1,+c_0^4$ and go to $\fbox{B}$.
%None of the other directions has been available and, thus, the algorithm has considered once all the directions from $D_0$. 
The rest of the directions in the cyclic order will not be available there. % and, thus, $D_0$ will get exhausted.
Finally, the algorithm will use $+c_0^1$ once and the token will go to $\fbox{H}$ which is the sink for $A_0$.

\begin{figure}[htbp] 
  \centering\includegraphics[width = 0.45\textwidth]{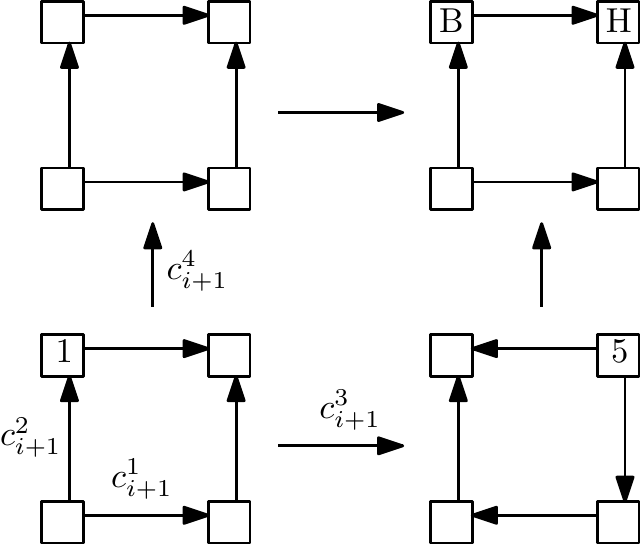}
\caption{The frame $F_3$ which is also the base case 4-AUSO $A_0$. The faces \fbox{1} and \fbox{5} are 
explicitly labeled since the sink $s_i$ of $A_i$ will be found (for the first time; Step (3) in the step-by-step analysis) in one of them.}
\label{fig:cunn_F3}
\end{figure}

%Let us give a short description of the behavior of Cunningham's rule on $A_{i+1}$. Firstly, the starting vertex $v_0 = v_0^{i+1} = \{c_0^2, \ldots, c_{i+1}^2\}$. Thus,
The list $L_{i+1}$ is defined based on the list $L_i$ in the following way (here $\cdot$ represents concatenation): 
\[L_{i+1} = L_i \cdot (+c_{i+1}^1,-c_{i+1}^2,+c_{i+1}^3,-c_{i+1}^1,+c_{i+1}^4,-c_{i+1}^3,+c_{i+1}^2,-c_{i+1}^4)\] 
This means that the directions from bundle $D_k$ have priority over the ones from bundle $D_{k'}$, if $k < k'$.
Let us denote by $IN$ the set of directions $D_i^+$ and let us denote by $OUT$ the set of directions from $D_{i+1}$.
The directions from $OUT$ are last in the list $L_{i+1}$. Let the starting vertex be $v_0^{i+1} = \{c_0^2, \ldots, c_{i+1}^2\}$.

We say that the directions from $IN$ have been \emph{exhausted} when the algorithm has already considered them and it is now the 
turn of directions from $OUT$ to be considered. 
Similarly, the directions from $OUT$ have been exhausted when the algorithm has made a full cycle of the list and the next
direction that will be considered is $L[1]$.
Note that frame $F_3$ has the property (w.r.t. $L_{i+1}$) that when $IN$ is exhausted at either \fbox{1} or \fbox{5} the token will 
go to \fbox{B}.

Let us now describe how the different frames are used to connect the different copies of $A_i$. Firstly, the starting vertex
$v_0^{i+1}$ is connected with frame $F_1$ and the sink vertex $s_i$ is connected with $F_3$. 
Note that the starting vertex is in \fbox{1}. 
Our goal is to force the algorithm to behave as follows:
It uses directions from $IN$, only in faces \fbox{1} and \fbox{5}, to follow path $P_i$ (in a projected way) and moves between these two faces,
using directions from $OUT$, whenever $IN$ is exhausted. 

Every time the token is in \fbox{1} and $IN$ gets exhausted, it takes a path 
$\fbox{1} \rightsquigarrow \fbox{5}$. There it arrives at a vertex where $OUT$ is exhausted. So, it uses a direction from $IN$. After one such move we change 
the connecting frame to $F_2$. This can be interpreted as an \emph{adversary argument} and can be implemented by Lemma~\ref{lem:product}.
Similarly, when the token is in \fbox{5} and $IN$ gets exhausted, it takes a path $\fbox{5} \rightsquigarrow \fbox{1}$. 
Like before, after that $OUT$ is exhausted and the algorithm will start considering the list from the beginning. After the first step on one of the directions from 
$IN$ we change the connecting frame ot $F_1$. 

Note that while moving between \fbox{1} and \fbox{5} the algorithm will take the path $P_i$ in a projected way. The adversary argument takes care only of vertices that are on this
projected path. For the rest of the vertices the frame does not really matter, since they will not be on the path $P_{i+1}$ that the algorithm will take.
We use the frame $F_3$ to connect those vertices.

The above procedure produces AUSO $A'_{i+1}$. The latter is acyclic because it has been produced by only using Lemma~\ref{lem:product} and 
all the connecting frames are acyclic. 
That is not the final AUSO as it remains to install the balance-AUSO in the face \fbox{B}. For this 
we use Lemma~\ref{lem:reorientation} to reorient the face \fbox{B} according to the balance-AUSO. The latter is 
the uniform AUSO that has its sink at $v_0^i$. This concludes the construction and the result is $A_{i+1}$. It is easy to see that 
the reorientation does not introduce any cycles: By the previous discussion, only vertices in \fbox{B} could be involved in cycles. Since the orientation in \fbox{B}
is acyclic, any possible cycle must involve coordinates from $C_{i+1}$. However,
in all the connecting frames \fbox{B} has only coordinate $c_{i+1}^1$ outgoing to the hypersink \fbox{H}. 
Thus, it is not possible for any vertex from \fbox{B} to be involved in any cycles and $A_{i+1}$ is acyclic.
Next we give a detailed step-by-step analysis of the behavior of the token on AUSO $A_{i+1}$.
\subparagraph{Step-by-step analysis.} 
The token is initially at the starting vertex $v_0^{i+1} = \{c_0^2, \ldots, c_{i+1}^2\}$. In detail:
\begin{enumerate}[(1)]
\item 
The token is at \fbox{1}. Directions from $IN$ will be used, since they have priority in $L_{i+1}$. % $b(d^{OUT}_{max}) = 0$. 
After some steps they will be exhausted. The connecting frame for $C_{i+1}$ will be $F_1$. 
The token then takes a path $\fbox{1} \rightsquigarrow \fbox{5}$. There,
the $OUT$ directions will be exhausted. The correctness of this argument can be checked in Figure~\ref{fig:cunnframes}.

\item 
The token is now in \fbox{5}. Again, directions from $IN$ will be used, since they have priority in $L_{i+1}$. 
After some steps they will be exhausted. The connecting frame for $C_{i+1}$ will be $F_2$. 
The token then takes a path $\fbox{5} \rightsquigarrow \fbox{10} \rightarrow \fbox{1}$.
There, the $OUT$ directions will be exhausted. The correctness of this argument can be checked in Figure~\ref{fig:cunnframes}.

\item The algorithm will keep looping between steps (1)-(2) until the token reaches 
a vertex $v_{s_i}$ such that $s(v_{s_i}) \cap C_i^+ = \emptyset$ (equivalently, $v_{s_i} \cap C_i^+ = s_i$).
That is a vertex that corresponds to the sink of $A_i$. 

\item The connecting frame for $v_{s_i}$ is $F_3$. From the arguments in (1) and (2) this vertex will be reached in either
\fbox{1} or \fbox{5}. Note that in both cases, the algorithm following $L_{i+1}$ will take the token to \fbox{B}
and the directions from $OUT$ will be exhausted there. 
The correctness of this argument can be checked in Figure~\ref{fig:cunn_F3}.

\item  The token now is at vertex $v_{s_i} \bot \fbox{B}$ and the algorithm will start checking at $L[1]$.
The token takes a path $v_{s_i} \bot \fbox{B} \rightsquigarrow v_0 \bot \fbox{B}$. The last vertex is the sink of \fbox{B}.

\item In the next step, the algorithm will use $+c_{i+1}^1$ to go to \fbox{H}. In there, the token is at vertex $v_0 \bot \fbox{H}$. The coordinates
from $C_{i+1}$ will remain incoming from now on. Thus, the token will take a path in \fbox{H} that is the same as $P_i$.
\end{enumerate}
With the above analysis, we have proved that the path $P_{i+1}$ will have length that is larger than twice the length of path $P_i$.
Therefore, we obtain the recursion $T(n+4) > 2T(n)$ which leads to the proof of Theorem~\ref{thm:cunn}.

\section{Exponential lower bound for Johnson's rule} \label{sec:johnson}
\begin{theorem} \label{thm:johnson}
There exist $n$-AUSO such that Johnson's rule, with a suitable starting vertex, takes a path of length 
at least $2^{n/4}$.
\end{theorem}

In this section we will prove the above theorem.
Firstly, let us define Johnson's least-recently basic rule. Consider that the algorithm runs on an $n$-AUSO.
It maintains a history function $h$ which is defined on all $2n$ directions. Let $v$ be the current vertex. 
Intuitively, the algorithm keeps the following history:  Say direction $d$ was used at step $x$ and let $-d$ be the 
opposite of that direction. Then, at step $x$ we have $h(-d) = x$; this will stay intact until $-d$ is used.
%(note that $-d$ cannot be used again before $+d$ is used). 
On the other hand, $h(d)$ will keep increasing to the current step until $-d$ is used. 

Formally, for a direction $d$, 
$h(d)$ is the last step number when $|d| \in v$ if $d$ is positive and the last step number when 
$|d| \notin v$ if $d$ is negative. Here, $v$ is the vertex on the path of the algorithm that corresponds to this last step and 
$|d|$ denotes the coordinate that corresponds to direction $d$. 
Note that ties are possible. We assume that 
those are broken lexicographically. 
The algorithm chooses from the set of available directions, direction $d$ which minimizes $h(d)$.
%An example of a run of this algorithm is included in Appendix~\ref{app:johnson_example}.

\subparagraph{An example run.} We describe an example run of Johnson's rule on the 4-AUSO $F_1$ that 
is given in Figure~\ref{fig:johnsonframes}. Note that $F_1$ will also serve as the base case $A_0$ for the inductive construction
(to be described after this example).
%Secondly, we will write in detail how the algorithm behaves on $A_0$, the base case of the construction, which is equivalent to $F_1$.
In the table below the first column is the step number. The second column is the vertex at the beginning of the corresponding step.
The third column is the direction chosen by the algorithm at that step. The rest of the columns show the history of each direction
after the step was performed. Below we use $\pm k$ to mean direction $\pm c_0^k$. The starting vertex is $\fbox{1}$.

\begin{center}
{\footnotesize
\begin{tabular}{|c|c|c||c|c|c|c|c|c|c|c|}
\hline 
Step &  vertex  & $d$  & h(+1) & h(+2) & h(+3) & h(+4) & h(-1) & h(-2) & h(-3) & h(-4) \\
\hline 
\hline
1 & \fbox{1}  & +1 & 1  & 0  & 0  & 0  & 1  & 1  & 1  & 1  \\ 
\hline 
2 & \fbox{2}  & +2 & 2  & 2  & 0  & 0  & 1  & 2  & 2  & 2  \\ 
\hline 
3 & \fbox{3} & +3  & 3  & 3  & 3  & 0  & 1  & 2  & 3  & 3  \\ 
\hline
4 & \fbox{4} & +4  & 4  & 4  & 4  & 4  & 1  & 2  & 3  & 4  \\ 
\hline
5 & \fbox{5} & -3  & 5  & 5  & 5  & 5  & 1  & 2  & 5  & 4  \\ 
\hline
6 & \fbox{R} & -2 & 6  & 6  & 5  & 6  & 1  & 6  & 6  & 4  \\ 
\hline
7 &\fbox{H} &  & 7  & 6  & 5  & 7  & 1  & 7  & 7  & 4   \\ 
\hline
\end{tabular} 
}
\end{center}

\subparagraph{The construction} is inductive.
Let $A_i$ denote the $i$th inductive step.
%For every inductive step $j$ we will describe the AUSO $A_j$ and the initial list $L_j$.
%The base case is $A_0$, a 4-AUSO; then, we add 4 more dimensions at every inductive step. We call  
%a bundle of coordinates $C_j = \{c^1_j, c^2_j, c^3_j, c^4_j\}$ the set of coordinates that were added at the $j$th step of induction.
%Thus, the AUSO $A_j$ is $4(j+1)$-dimensional.
%To this extent, let us define $C_j^+ = \bigcup_{k=0}^j C_k$, $|C_j^+| = 4(j+1)$.
The base case $A_0$ is the 4-dimensional AUSO $F_1$, shown in Figure~\ref{fig:johnsonframes}.  The initial set of coordinates 
is $C_0 = \{c^1_0, c^2_0, c^3_0, c^4_0\}$. % and the initial list $L_0 = (c^1_0, c^2_0, c^3_0, c^4_0)$. 
The starting vertex is $v_0 = \emptyset$, which is at the vertex labeled $\fbox{1}$ in the figure.
As discussed in the previous paragraph, the algorithm will go over directions $(+c^1_0, +c^2_0, +c^3_0, +c^4_0, -c^3_0, -c^2_0)$ and will 
find the sink at $\{c^1_0,c^4_0\}$. %\AT{Use term Hypersink?}
Let us now describe how to construct AUSO $A_{i+1}$ from AUSO $A_i$. Every inductive step adds 4 dimensions. 
As before, $C_j^+ = \bigcup_{k=0}^j C_k$. %; and $D_j$ denotes the set of corresponding directions.
Let the new bundle of coordinates be $C_{i+1} = \{c^1_{i+1}, c^2_{i+1}, c^3_{i+1}, c^4_{i+1}\}$.
We take 16 copies of $A_i$ and connect them with the 4-AUSO $F_1$ and $F_2$ that 
appear in Figure~\ref{fig:johnsonframes}. 
In this section, a reset-AUSO (thus, the $\fbox{R}$ in the figure) will take the role of the balancing gadget. 

\begin{figure}[htbp] 
  \centering\includegraphics[width = \textwidth]{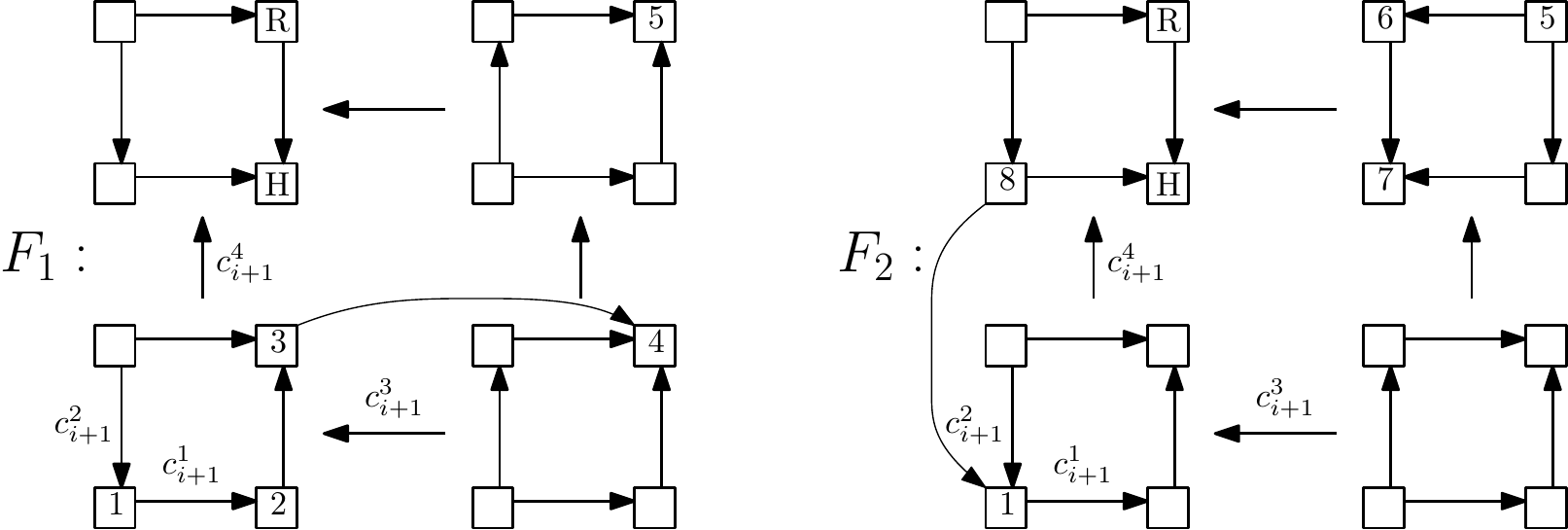}
\caption{The orientations $F_1$ and $F_2$, used as connecting frames, are given in this figure. 
}
\label{fig:johnsonframes}
\end{figure}
%First, let us define $L_{i+1}$: 
%\[L_{i+1} = L_i \cdot (c^1_{i+1}, c^2_{i+1}, c^3_{i+1}, c^4_{i+1})\]
%where $\cdot$ represents concatenation and $L_i$ is the initial list of inductive step $i$. 

Let us define what we mean by \emph{lexicographic order} here: $+c_j^k$ comes before $-c_j^{k'}$ for any $k$ and $k'$
(positive before negative in the same bundle);
$+c_j^k$ comes before $+c_j^{k'}$ and $-c_j^k$ comes before $-c_j^{k'}$ if $k < k'$
(in the same bundle smaller-index directions come first). Finally,
$d_j^k$ comes before $d_{j'}^{k'}$ for any $k$ and $k'$ and positive/negative sign, if $j < j'$
(bundles with smaller index come first).
%$c_j^k$ comes before $c_{j'}^{k'}$ for any $k$ and $k'$ if $j < j'$ and $c_j^k$ comes before $c_j^{k'}$ if $k < k'$.
%Moreover, the notation $c \prec c'$ means that $c$ comes before $c'$ in the current priority list $L$.

The starting vertex will be $v_0 = \emptyset$ for every inductive step. 
%Then, the run of the algorithm on $A_{i+1}$ defines a path $P_{i+1}$ that starts at vertex $\emptyset$ and walks all the way to the sink. 
Then, every positive direction $+c$ initially has $h(+c)=0$ (until $+c$ is used by the algorithm).
At step number 1, one of the positive directions will be used at which point every 
negative direction $-c$ has $h(-c) = 1$.

With this construction we want to force the algorithm to the following behavior: It starts at \fbox{1} using all the positive directions  
in lexicographic order. Then it is in \fbox{5}, where it will use all the negative directions in lexicographic order. This will continue until 
the sink of $A_i$ has been reached. It follows that directions from $D_i^+$ are only used when the algorithm is in \fbox{1} or \fbox{5}, 
before the sink of $A_i$ has been discovered. We will later show that this is the case.

The construction is considered as an adversary argument (similarly to Section~\ref{sec:cunn}). 
Firstly, the starting vertex of $A_i$ and the sink of $A_i$ both use $F_1$ as the connecting frame. % \AT{removed the 'associate' sentence}
%Then, we associate $F_1$ to $\fbox{1}$ and $F_2$ to $\fbox{5}$.
Every time the algorithm is in $\fbox{1}$ and uses a direction from $D_i^+$, we change (or keep) the connecting frame to $F_1$.
%Similarly, for $\fbox{5}$ and $F_2$ (except if the algorithm has reached the sink of $A_i$).
Similarly, when the algorithm arrives in $\fbox{5}$ and uses a direction from $D_i^+$, 
we change the connecting frame to $F_2$ (except if the algorithm has reached the sink of $A_i$).
%Note that both frames $F_1$ and $F_2$ are such that the sink is at the same vertex \fbox{H} and that for both of them $\fbox{R}$
%is a facet sink with only one outgoing edge towards $\fbox{H}$. Like in the previous section we consider the construction as an
%adversary argument. Firstly, the starting vertex of $A_i$ and the sink of $A_i$ both use $F_1$ as the connecting frame.
%Then, we associate $F_1$ to hypervertex $\fbox{1}$ and $F_2$ to hypervertex $\fbox{5}$. Every time the token arrives at 
%$\fbox{1}$ and the algorithm uses a direction from $D_{i}^+$ we make sure that the frame is $F_1$. Every time the token arrives at 
%$\fbox{5}$ and the algorithm uses a direction from $D_{i}^+$ we make sure that the frame is $F_2$. We will show later 
%that directions from  $D_{i}^+$ will only be used in these two hypervertices.
This operation is consistent with Lemma~\ref{lem:product}: Every vertex is connected with the corresponding frame.
The result of this operation is $A'_{i+1}$ which is not the final AUSO.
The final step for the construction of $A_{i+1}$ is to use Lemma~\ref{lem:reorientation} to embed 
a reset-AUSO in the face
$\fbox{R}$.

\subparagraph{Constructing the reset-AUSO.}
This construction is also inductive. Let $R_i$ denote the $i$th inductive step. 
Then, $R_i$ is a $4i$-AUSO (whereas $A_i$ is a $4(i+1)$-AUSO). Since $R_i$ is placed in a face of $A_i$, 
the coordinates it spans are the ones in $C_{i-1}^+$.
The final step in the construction of $A_{i+1}$ is to embed the reset-AUSO $R_{i+1}$ in the face
$\fbox{R} = F(C_i^+, \{c^1_{i+1}, c^2_{i+1}, c^4_{i+1}\} )$ of $A'_{i+1}$. 

$R_0$ is just a 0-AUSO, i.e. a vertex; or equivalently, there is no reset-AUSO for $A_0$. Then, let $R_1$ be the base case; this 4-AUSO is 
depicted in Figure~\ref{fig:johnsonreset} below.
Then, to construct $R_{i+1}$ from $R_i$ we take 16 copies of $R_i$ which we connect with two frames. 
For every vertex in $R_i$ which is not the sink, we use for connecting frame the uniform 4-AUSO which has its sink at $\{c^1_{i+1},c^4_{i+1}\}$.
For the sink of $R_i$ we use for a frame the 4-AUSO $R_1$. 
The intuition behind reset-AUSO $R_{i+1}$ is summed up by the following property:
It has a path from vertex
$\{c_0^1,c_0^4, \ldots, c_{i}^1,c_{i}^4\}$ to the vertex $\emptyset$ such that every vertex on this path has only one outgoing edge
and the path goes through the negative directions in lexicographic order: $(-c_0^1, -c_0^4,\ldots, -c_i^4, -c_i^4)$. 

\begin{figure}[htbp] 
  \centering\includegraphics[width = 0.4\textwidth]{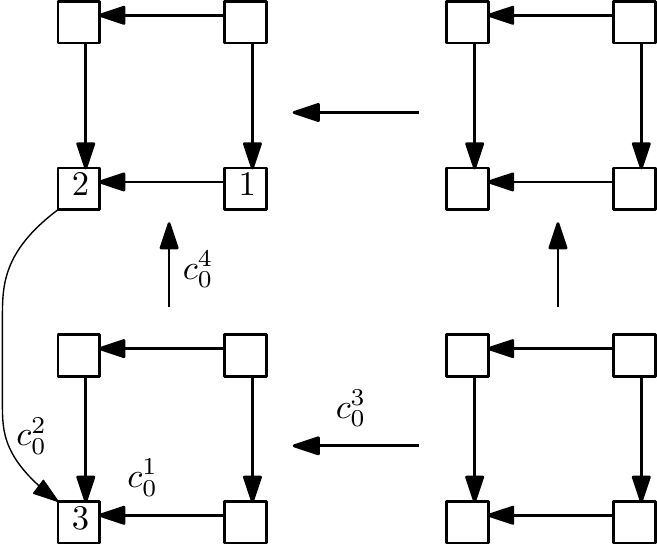}
\caption{The reset-AUSO $R_1$. The numbers indicate the path the token will take when resetting $C_0$ (``resetting'' definition to follow).}
\label{fig:johnsonreset}
\end{figure}

The final step for the construction of $A_{i+1}$ is to use Lemma~\ref{lem:reorientation} to embed $R_{i+1}$ to the face
$\fbox{R}$. This concludes the construction of $A_{i+1}$. 
Note that the last step does not introduce any cycles:
$R_{i+1}$ is acyclic and, thus, the only possible cycles would involve the edges on coordinates external to \fbox{R},
i.e. the coordinates from $C_{i+1}$. However, the hypervertex \fbox{R} has only one outgoing edge on 
$-c_{i+1}^2$ which leads to the hypersink \fbox{H}; the latter does not have any outgoing edges on $C_{i+1}$ and, so, 
a cycle is not possible. We conclude that the constructed AUSO is acyclic.

\subparagraph{The behavior of Johnson's rule} on the AUSO constructed as above will be described here. %Let $n$ represent the dimension of this AUSO. 
Firstly,  we define the tools that we
are going to use for this analysis.
% At any point below $v$ is the \emph{current vertex} (the one where the algorithm currently is).
%In addition, with $L$ we denote the current ordered list of coordinates (formed after the algorithm arrived at the current vertex $v$).

%\paragraph*{Tokens and active bundles.} 
Similarly to the previous section, consider a token $t$. That is a token that starts at the initial vertex $v_0$ and moves according to the directions that the algorithm 
%\AT{Fix notation of token to be the same like before}
chooses. With slight abuse of notation we also use $t$ to refer to the vertex where the token currently lies on. Moreover, we write 
$t_j$ to mean the set $t \cap C_j^+$; that is, the projection of the vertex $t$ to the set of coordinates $C_j^+$.
Since $t_j \subseteq t$, we call $t_j$ a subtoken.

We say that a coordinate bundle $C_j$ is \emph{active} when $s(t) \cap C_j \neq \emptyset$.
Otherwise, we say that $C_j$ is inactive. Note that for both the 4-dimensional frames that we have used, the sink is at the same vertex. 
This implies that $C_j$ is inactive if and only if token $t$ is such that 
$t \cap C_j = \{c^1_j, c^4_j\}$.
Moreover, we say that token $t$ is in $\fbox{1}_j$ to mean that $t \cap C_j = \emptyset$; $t$ is in $\fbox{5}_j$
when $t \cap C_j = C_j$ and similarly for the rest of the faces \fbox{$\cdot$} from Figure~\ref{fig:johnsonframes}.

For each subtoken $t_j$, we say that it has reached \emph{its sink} when all bundles in $C_j^+$ are inactive.
This means that
$t_j = \{c_0^1, c_0^4, \ldots, c_j^1, c_j^4\}$.
%Moreover, we say that $t_j$ is at \emph{its initial vertex} when $t_j = \emptyset$.
%\AT{Check if we actually need initial vertex.}
\emph{Resetting} $C_j^+$ is a process that happens when subtoken $t_j$ is at its sink:
Token $t$ moves
from a vertex where $t \cap C_j^+ = \{c_0^1, c_0^4, \ldots, c_j^1, c_j^4\}$ 
to a vertex where $t \cap C_j^+ = \emptyset$.
Moreover, we say that $C_j^+$ is \emph{resettable} when:  %\AT{Change this to $j'+1$? Do according to lemma 7}
\begin{itemize}
\item $t_j$ is at its sink. %, i.e. all bundles in $C_j^+$ are inactive. 
\item $h(-c_{j'}^1) < h(-c_{j'}^4) < h(-c_{j'+1}^2)$, for every $0 \leq j' \leq j$. %, where $c' \in C_{j+1}$. %$c' \in C_{j''}$ and $j'' > j$.  %$c' \in C_{j'+1}$. 
%\AT{Make sure what is the right definition. It's possible that 
%larger index coordinates have priority but they get used in the right order because of the reset AUSO construction.}
\end{itemize}
The first bullet in the definition above equivalently means that all bundles in $C_j^+$ are inactive. % bundle $C_{j'}$ is inactive, for every $0 \leq j' \leq j$.
%The second bullet dictates that, for every $0\leq j' \leq j$, $-c_{j'}^1$ has priority over $-c_{j'}^4$ 
%(if their history is equal this is due to the lexicographic order)
%and that both of them have priority over direction $-c_{j+1}^2$. 
Resetting $C_j^+$ is a process that takes place when (and only when) 
token $t$ is in the reset-AUSO in $\fbox{R}_{j+1}$. 
For now assume that $C_j^+$ will be reset only when it is resettable; we will prove this later (with Lemma~\ref{lem:resettable}). 
Thus, $t_j$ is on its sink when the resetting process starts.
%Token $t$ moves from a vertex $t \cap C_j^+ = \{c_0^1, c_0^4, \ldots, c_j^1, c_j^4\}$ to a vertex such that
%$t \cap C_j^+ = \emptyset$. 
The second bullet of the definition of resettable ensures that token $t$ will not go out of $\fbox{R}_{j+1}$ before
$C_j^+$ has been reset.
During the reset of  $C_j^+$, token $t$ will
go over negative directions from $D_j^+$ in this order: 
$(-c^1_0, -c^4_0, \ldots, -c^1_j, -c^4_j)$. This is because of the construction of the reset-AUSO $R_{j+1}$; 
the algorithm has no other choice.

We are ready to state the following lemma. It dictates the first steps that the token will take, either at 
%The next lemma dictates the first steps that the token will take, either at 
the beginning of the execution of the algorithm
or after a reset. Since the algorithm is deterministic, this defines the path that the token will take to the sink.

\begin{lemma} \label{lem:at0allpos}
%Let $t \cap C_j = \emptyset$. consecutively all positive from $C_j$.
Let $t \cap C_j^+ = \emptyset$. Then, all the positive directions from
$C_j^+$ will be used in the lexicographic order: $(+c_0^1,+c_0^2,+c_0^3,+c_0^4,\ldots,+c_j^1,+c_j^2,+c_j^3,+c_j^4)$.
\end{lemma}
\begin{proof}

The token $t$ is at a vertex such that $t \cap C_j^+ = \emptyset$. Note that the connecting frame is $F_1$ for all 
the bundles in $C_j^+$. There are two cases. The first case is when the algorithm 
is just starting at the initial vertex $v_0$. Since $v_0 = \emptyset$, all positive directions have history 0 and they will be used in the lexicographic order.

The other case is after $C_j^+$ has been reset. By the definition of resetting, we have that $(-c_0^1, -c_0^4, \ldots, -c_j^1, -c_j^4)$
have been lastly used. Thus, for every $0 \leq k \leq j$ we have that $h(+c_k^1) < h(+c_k^4)$,
$h(+c_k^2) < h(+c_k^1)$ and $h(+c_k^3) <  h(+c_k^1)$.
In addition, $h(+c_k^1) < h(+c_{k'}^1)$ for all $0 \leq k < k' \leq j$. 
%$h(+c_{k'}^2) < h(+c_{k''}^3)$ for all $0 \leq k' < k''$. 
Since the connecting frame is $F_1$ for all the bundles in $C_j^+$, we have that direction $+c_k^2$ is only available after 
$+c_k^1$ has been used and direction $+c_k^3$ is only available after $+c_k^1$ and $+c_k^2$ has been used.
So, the direction $+c_0^1$ is the first one that will be used. Then, $+c_0^2$
becomes available and will be used, then $+c_0^3$ becomes available and will be used and, finally,  $+c_0^4$ will be used. 
Then, $+c_1^1$ will be used and so on and so forth. We conclude that all the positive directions from $C_j^+$ 
will be used in the lexicographic order.
\end{proof}

The above lemma defines the path that token $t$ will follow until subtoken $t_j$ reaches its sink. For every bundle $C_j$, the 
positive directions are used in lexicographic order $(+c_j^1,+c_j^2,+c_j^3,+c_j^4)$ and the token $t$ goes to $\fbox{5}_j$.
After this, we have $h(-c_j^1) < h(-c_j^2) < h(-c_j^3) < h(-c_j^4) < h(d)$ for any positive $d$ from $D_j$. 
Then, some directions from $D_{j-1}^+$ will be used and the frame for $C_j$ will change to $F_2$. When it is the turn of the negative 
directions from $D_j$ to be used they will be used consecutively and 
in lexicographic order $(-c_j^1,-c_j^2,-c_j^3,-c_j^4)$; that is, assuming $t_{j-1}$ has not reached its sink yet. 
Otherwise, the connecting frame for $C_j$ would be $F_1$ and the directions $-c_j^1$ and $-c_j^4$ would not be available.
% would not 
%be available.

After this, $t$ will be in $\fbox{1}_j$ and, moreover, $h(+c_j^1) < h(+c_j^2) < h(+c_j^3) < h(+c_j^4) < h(d)$ for any negative $d$ from $D_j$.
There, after some direction from $D_{j-1}^+$ is used, the connecting frame for $C_j$ will be $F_1$. 
When it is the turn of the positive 
directions from $D_j$ to be used they will be used consecutively and 
in lexicographic order $(+c_j^1,+c_j^2,+c_j^3,+c_j^4)$ and the token $t$ will go back to $\fbox{5}_j$. 
We can conclude that $t$ will keep moving from $\fbox{1}_j$ to $\fbox{5}_j$ and back until subtoken $t_{j-1}$ reaches its sink. 
The next corollary follows from this discussion.
%Let us phrase the next corollary, which follows from the above discussion.

\begin{corollary} \label{cor:johnson}
Let $C_{j+1}$ be active. 
\begin{enumerate}
\item If $t$ is in $\fbox{1}_{j+1}$, the positive directions from $D_{j+1}$  will be used (when it is their turn) 
consecutively as $(+c_{j+1}^1,+c_{j+1}^2,+c_{j+1}^3,+c_{j+1}^4)$.
\item If $t$ is in $\fbox{5}_{j+1}$ and 
subtoken $t_{j}$ has not reached its sink, then the negative directions from $D_{j+1}$  will, similarly, be used (when it is their turn) 
consecutively as $(-c_{j+1}^1,-c_{j+1}^2,-c_{j+1}^3,-c_{j+1}^4)$. 
\end{enumerate}
It follows that directions from $D_{j}^+$ are only used when $t$ is in $\fbox{1}_{j+1}$ or in $\fbox{5}_{j+1}$.
\end{corollary} 

%OLD COROLLARIES
%\begin{corollary} \label{cor:usepositive}
%Let $j > 0$ and token $t$ be at $\fbox{1}_j$. When the positive directions from $D_j$ are to be used they will be used consecutively and in lexicographic order,
%$(+c_j^1,+c_j^2,+c_j^3,+c_j^4)$. 
%\end{corollary}
%\AT{Make this corollaries shorter}
%\begin{corollary} \label{cor:usenegative}
%Let $j>0$, $t$ be at $\fbox{5}_j$ and subtoken $t_{j-1}$ have not reached its sink.
%When the negative directions from $D_j$ are to be used they will be used consecutively and in lexicographic order,
%$(-c_j^1,-c_j^2,-c_j^3,-c_j^4)$. 
%\end{corollary}

%\begin{corollary} \label{cor:onlyin1and5}
%Let $j>0$. Directions from $D_{j-1}^+$ are only used when $t$ is in $\fbox{1}_j$ or in $\fbox{5}_j$.
%\end{corollary}

The next lemma is the last ingredient needed for the proof of Theorem~\ref{thm:johnson}.
\begin{lemma} \label{lem:resettable}
Let $C_{j+1}$ be active. When $t_j$ reaches its sink, $C_j^+$ is resettable.
\end{lemma}
\begin{proof}
By definition, when $t_j$ reaches its sink all bundles in $C_j^+$ are inactive. 

We prove the statement by induction on $j$. For the base case $j=0$, we explicitly give the directions 
from $D_1^+$ in the order they are used by the algorithm until $t_0$ reaches its sink: 
$(+c_0^1,+c_0^2,+c_0^3,+c_0^4,+c_1^1,+c_1^2,+c_1^3,+c_1^4,-c_0^3,-c_0^2)$.
It follows that $C_0^+ = C_0$ is resettable when $t_0$ reaches its sink (while $C_1$ is still active).

Now consider any $k$, $1 < k \leq j$. Let $C_{k+1}$ be active; 
we will prove that $C_k^+$ is resettable when $t_k$
reaches its sink.
%When $t_k$ reaches its sink it is already the case that
%$t$ is in $\fbox{H}_k$. 
Consider %the situation before $t$ entered $\fbox{H}_k$. Specifically, 
the last step when
$t$ was in $\fbox{1}_k$. %before subtoken $t_{k-1}$ reached its sink. Regardless of 
%Regardless of where is $t$ when $t_{k-1}$ reaches its sink,   
The algorithm will use directions $(+c_k^1,+c_k^2,+c_k^3,+c_k^4)$,
by Corollary~\ref{cor:johnson},
%Note that after this the positive directions from $D_{k+1}$ were used by appl
and token $t$ will go to $\fbox{5}_k$. 
Note that this is regardless of whether subtoken $t_{k-1}$ has reached its sink.
%It does not matter whether the subtoken $t_{k-1}$ reached its sink 
%when $t$ was $\fbox{1}_k$ or when $t$  was $\fbox{5}_k$.
The result is that $h(-c_k^1)<h(-c_k^2)<h(-c_k^3)<h(-c_k^4)$. The subtoken $t_{k-1}$ reaches its sink
either when $t$ is in $\fbox{1}_k$ or when $t$ is in $\fbox{5}_k$, by Corollary~\ref{cor:johnson}.
%either before the positive directions from $D_k$ were used (in which case $t$ was in $\fbox{1}_k$) or after 
%they were used (in which case $t$ was in $\fbox{5}_k$).

Afterwards, 
the algorithm uses direction $-c_k^3$ and $t$ goes to $\fbox{R}_k$
($C_k$ is still active). Therein, $C_{k-1}^+$ will get reset since it is resettable
by the inductive hypothesis. Following, 
the algorithm uses $-c_k^2$, $t$ goes to $\fbox{H}_k$,
and $C_k$ becomes inactive. So, we still have $h(-c_k^1)< h(-c_k^4)$. In addition, 
since $C_{k-1}^+$ just got reset we have that $t\cap C_{k-1}^+ = \emptyset$. 
By Lemma~\ref{lem:at0allpos}, all the positive directions from $D_{k-1}^+$ will be used. 
Since $C_{k+1}$ is still active, the positive directions from $D_{k+1}$ will be used at least once; 
this includes $+c_{j+1}^2$. However,  $C_k$ is already inactive and no directions from $D_k$ will be used.
Thus, we have $h(-c_k^1)< h(-c_k^4) < h(-c_{k+1}^2)$.
\end{proof}

%\paragraph*{Putting it all together.}
Now consider the path of the token $t$ from $v_0$ to the sink of $A_{i+1}$. This AUSO is of 
dimension $n = 4(i+2)$.
%Consider the outermost bundle is $C_{n'}$; hence, the AUSO is of dimension $n= (4n'+1)$. We have proved 
%that token $t$ will be moving back and forth between $\fbox{1}_{n'}$ and $\fbox{5}_{n'}$ until $t_{n'-1}$ reaches its sink.
By Corollary~\ref{cor:johnson}, $t$ will be moving back and forth between $\fbox{1}$ and $\fbox{5}$ until $t_{i}$ reaches its sink.
If we project the path that $t$ has taken until this point on coordinates from $C_i^+$ it will be the same as $P_i$.
When $t_i$ reaches its sink, $C_{i}^+$ is resettable (by Lemma~\ref{lem:resettable}) and, when $t$ enters $\fbox{R}$, $C_{i}^+$ will be reset. 
Following, $t$ enters $\fbox{H}$ at vertex $v_0 \bot \fbox{H}$. %such that $t\cap C_{n'-1}^+ = \emptyset$. 
%\AT{Is there anywhere else to use the $\bot$ notation?}

%When the algorithm started at $v_0 = \emptyset$, all the positive directions from $D_{i+1}^+$ where used in lexicographic order. 
%Since it is deterministic, this defines completely the behavior of $t$. %path that $t$ will take.
Consider the path $P'$ that $t$ will follow from $v_0\bot \fbox{H}$ until $t_{i}$ reaches its sink. 
%but projected on the coordinates from $C_{i}^+$.
%Also, observe that the moves the token does on directions from $D_{i}^+$ are independent from the ones it does
%on directions from $D_{i+1}$. 
Since every vertex in $\fbox{H}$ has only incoming edges on $C_{i+1}$, the coordinates that the path uses are 
from the set $C_i^+$.
From Lemma~\ref{lem:at0allpos}, we know that when token $t$ is such that $t\cap C_{i}^+ = \emptyset$, all 
the positive directions from $D_{i}^+$ will be used in the lexicographic order. 
%At this point $C_{i+1}$ is inactive and $t$ is in $\fbox{H}$. Therein, it will follow the same path $P$ to the global sink.
Therefore, the path $P'$ will be the same as $P_i$ when $t_i$ reaches its sink, at which point $t$ will have 
reached the global sink of $A_{i+1}$.

Let $T(n)$ denote the length of the path that token $t$ will take from $v_0$ until it reaches the global sink on a $n$-AUSO.
With the above analysis, we have shown that the recursion  % that by adding  4 dimensions the length of this path doubles. In particular, we 
$T(n+4) > 2T(n)$ holds. This recursion leads to the proof of Theorem~\ref{thm:johnson}.

\section{Exponential lower bound for Zadeh's Rule} \label{sec:zadeh}

\begin{theorem} \label{thm:zadeh}
There exist $n$-AUSO such that Zadeh's rule, with a suitable starting vertex and tie-breaking rule, takes a path of length 
at least $2^{n/6}$.
\end{theorem}

In this section we will prove the above theorem.
Firstly, let us define formally Zadeh's least entered rule.
Consider that the algorithm runs on an $n$-AUSO. It maintains a history function $h$ which is defined 
on all $2n$ directions.
Given a direction $d$, $h(d)$ is the number of times the direction $d$ has been used. At the beginning $h(d)=0$, for all $d$.
At every step the algorithm picks one direction from the set of available ones that minimizes the history function.
In addition, there is a tie-breaking rule: this is an ordering of the directions and is invoked only in case 
more than one have the minimum history size. As we already mentioned in Section~\ref{sec:intro},
our lower bound construction will have the \emph{simplest} possible tie-breaking rule, an ordered list which will be given explicitly. This is in  contrast 
to the lower bounds by Friedmann \cite{Friedmann11}.

Secondly, let us define some tools that we will use for the analysis of the algorithm.
We have a \emph{balance} function $b$, which is also defined on all the $2n$ directions. 
Let $d_{max}$ be the most used direction; then, $b(d) = h(d_{max}) - h(d)$. This means that direction $d$ 
has been used $b(d)$ less times compared to $d_{max}$. 
%(This can also be defined on a subset of all the directions; then $d_{max}$ is 
%also defined to be in this subset). 
We say that a direction $d$ is \emph{imbalanced} when $b(d) > 0$ and that a set of directions 
$D$ is balanced when $b(d)=0$, for all $d\in D$.
We also define a balance function on any subset of directions: Given set $D$ we define $b(D,d)$ to be the balance 
of direction $d$ w.r.t. the directions from $D$, i.e. the defining coordinate is now $d_{max} \in D$.

Furthermore, we define the concept of \emph{saturation}. This is with regards to the history and the current vertex 
in the algorithm run. 
Given a set of directions $D$ and a vertex $v$ we say that $v$ is $D$-saturated when
for every $d \in D$ with $b(d) > 0$, the direction $d$ is not available for $v$. It follows that if α vertex $v$ 
is $D$-balanced, then $v$ is also $D$-saturated.

\vspace{-0.3cm}
\subparagraph{The construction } is inductive. 
Let $A_i$ denote the $i$th inductive step. The base case, $A_0$, 
is a 6-AUSO. This will be described later with Figure~\ref{fig:zadehBC}.
%Due to the lack of space we do not define it here, but  
%we will mention and utilize some of its properties. A complete description can be found in Appendix~\ref{app:zadeh_BC}.
Following, every inductive step adds 6 new dimensions.
As before the bundle of coordinates $C_i$ is the one added at the $i$th step of induction
(and $C_0$ is the bundle for the base case). Also, with $D_i$ we denote the directions that correspond to $C_i$
and, similarly, for $D_i^+$.
Thus, the AUSO $A_i$ is $6(i+1)$-dimensional and the coordinates that describe it are in the set $C_i^+$.
For each $A_i$, the starting vertex is $v_0^i = \{c_0^2, \ldots, c_i^2\}$. 

%In the analysis of the lower bound we talk of a token. That is initially placed at the starting vertex and  \AT{Cut off this para.}
%moves according to the directions that the algorithm chooses. 
%Let us denote with $P_i$ the path that the token takes from $v_0^i$ to the sink $s_i$ of $A_i$ (that is, when $A_i$ is given as input 
%to the algorithm) with respect to the tie-breaking rule $T_i$. The latter will be defined later in the text.

Let us now describe how to construct AUSO $A_{i+1}$ from AUSO $A_i$. 
%Every inductive step adds 6 dimensions. 
Let the new bundle of coordinates be $C_{i+1}$. % = \{c^1_{i+1}, c^2_{i+1}, c^3_{i+1}, c^4_{i+1}, c^5_{i+1}, c^6_{i+1}\}$.
We call $IN$ the set of directions $D_i^+$ and $OUT$ the set of directions $D_{i+1}$.
% that correspond to the coordinates in $C_i^+$; we call $OUT$ the set of directions that correspond to 
%the coordinates in $C_{i+1}$. 
At every inductive step the tie-breaking rule will be formed such that the directions from $IN$ have priority over the ones from $OUT$.
Thus, directions $D_k$ have priority over the ones from $D_{k'}$, if $k <k'$.
%For simplicity, we write number $\pm k$ to mean direction $\pm c_{i+1}^k$.

The starting vertex for $A_{i+1}$ is $v_0 = v_0^{i+1}=\{c_0^2, \ldots, c^2_{i+1}\}$. % and, similarly, $P_{i+1}$ \AT{define this stuff based on $P_i$.}
%Assume that $P_i$ (the path the token takes in $A_i$) is known to us. 
Similarly to the previous sections, 
the first part of the construction is interpreted as an adversary argument.
To construct $A'_{i+1}$ (not the final construction)
we take $2^6 = 64$ copies of $A_i$ and use three different 6-AUSO as connecting frames, utilizing Lemma~\ref{lem:product}. 
The frames are given in Figures~\ref{fig:zadeh_frames} and~\ref{fig:zadehF3}.
%For vertices that are not on $P_i$ it does not matter which frame we use. %, so let us use a forward uniform orientation.
For each vertex on the path of the algorithm we choose the connecting frame according to the following rule: 
\begin{enumerate}[(1)]
\item We connect with $F_1$ vertices that are not $IN$-saturated;
\item We connect with $F_2$ vertices that are $IN$-saturated;
\item For the sink $s_{i}$ of $A_i$ we use $F_3$.
\end{enumerate}
Note that the connecting frame is not important for vertices that are not on the path of the algorithm. For the sake of completeness
let us connect these vertices using $F_3$ as the connecting frame.
The latter is a 6-AUSO that has the same path $\rbox{1} \rightsquigarrow \rbox{12}$ as $F_1$, has its sink at \rbox{12} 
and all other edges are forward. It will is described pictorially in Figure~\ref{fig:zadehF3}.

The result of this operation is $A'_{i+1}$. The latter is acyclic since it arises by 
applying only Lemma~\ref{lem:product}. To obtain our final construction it 
remains to perform one reorientation. 
Namely, the balance-AUSO %(the gadget we discussed in the previous section)
will be embedded in the face $\fbox{B}= F(C_i^+, \{c^1_{i+1},c^2_{i+1},c^3_{i+1}\} )$,
shown in Figure~\ref{fig:zadeh_frames}.
The balance-AUSO is the uniform $6(i+1)$-AUSO which has its sink at the vertex $v_0^i$.
Formally, the outmap of vertex $v = v_0 \bot \fbox{B}$ is such that $s(v)\cap C_i^+ = \emptyset$.
The result of that reorientation will be the final AUSO $A_{i+1}$. 
Firstly, 
let us argue that $A_{i+1}$ is acyclic.

\begin{lemma} 
$A_{i+1}$ is acyclic. 
\end{lemma}
\begin{proof}
We already mentioned above that $A'_{i+1}$ is acyclic; it remains to argue that after the reorientation of \fbox{B}
doe not introduce any cycles. Note that \fbox{B} has only forward edges incident 
on the coordinates from $C_{i+1}$
in all three frames. From \fbox{B} there are paths to
vertices in faces from the set $S_1 = \{\rbox{1}, \ldots, \rbox{12}\}$; those are the only vertices reachable from $\fbox{B}$ that have 
incident backward edges on coordinates from $C_{i+1}$. However, in none of the frames there is path from any 
faces in set $S_1$ to some face from the set $S_2 = \{\fbox{1}, \ldots, \fbox{12}\} \cup \{F(C_i^+, \emptyset) \cup F(C_i^+, \{c_{i+1}^2, c_{i+1}^3\})\}$.
The set $S_2$ contains all the faces which have vertices with paths to \fbox{B} (in all three faces). Therefore, there is no cycle involving \fbox{B} and, then, 
$A_{i+1}$ is acyclic.
\end{proof}

Secondly, let us give the figure that describes the frames $F_1$ and $F_2$ (Figure~\ref{fig:zadeh_frames}).

\newpage

 \begin{figure}[htbp]
   \centering\includegraphics[width = \textwidth]{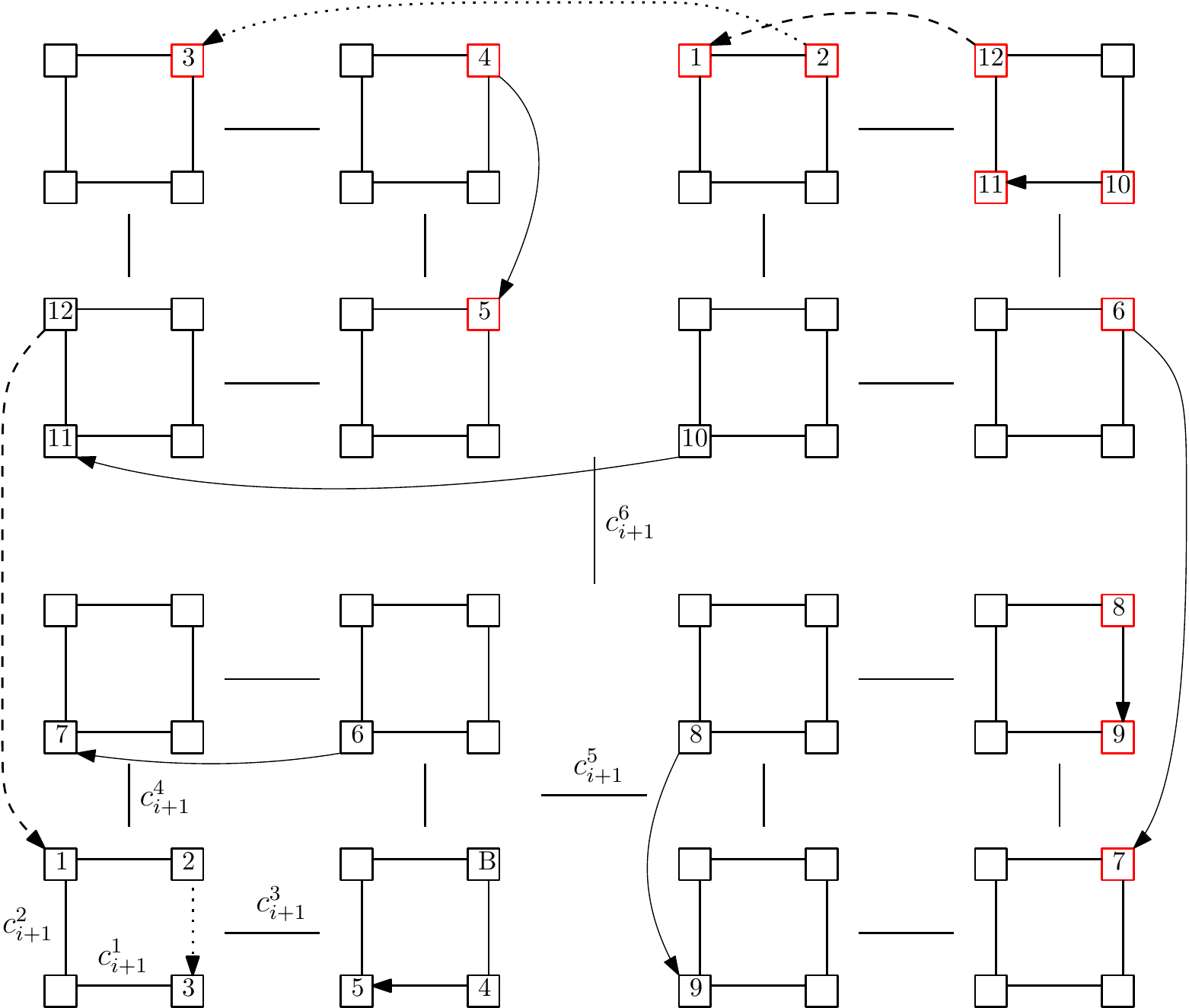}
 \caption{Both orientations $F_1$ and $F_2$ are given in this figure. For simplicity, only the backward edges 
 are explicitly oriented; every other edge is forward.
 The frame $F_1$ includes the \emph{dashed} backward edges but not the dotted ones;
 $F_2$ includes the \emph{dotted} backward edges but not the dashed ones. The solid backward edges are included in both frames.
 }
 \label{fig:zadeh_frames}
 %  cave are given in the figure. The largest ball that fits in that cave has 
  % radius 1.}
 \end{figure}
 
 In reference to Figure~\ref{fig:zadeh_frames}, let us present the intuitive idea of this lower bound: 
 The token will walk (in a projected way) along path $P_i$, while walking 
 between \fbox{1} and \fbox{12}. Then, it will go back to the start of $P_i$ in the balance-AUSO \fbox{B}. Then, it will walk the path 
  $P_i$ once again (in a projected way), while walking between \rbox{1} and \rbox{12}.
 
The frame $F_3$, which is used to connect the vertices that correspond to the sink $s_i$ of $A_i$, is presented in the 
next figure (Figure~\ref{fig:zadehF3}). 
Note that the sink of the frame is at \rbox{12}, which also places the global sink of $A_{i+1}$ in \rbox{12}.
The frame $F_3$ has two useful properties. The first 
is that it has only forward edges attached at \fbox{1}. This will take the token from there to \fbox{B} (w.r.t. $T_{i+1}$). The second is that 
there exists 
the same path $\rbox{1} \rightsquigarrow \rbox{12}$ as in $F_2$ (but the sink is at \rbox{12} for $F_3$). 

\newpage

\begin{figure}[htbp] 
 \centering\includegraphics[width = \textwidth]{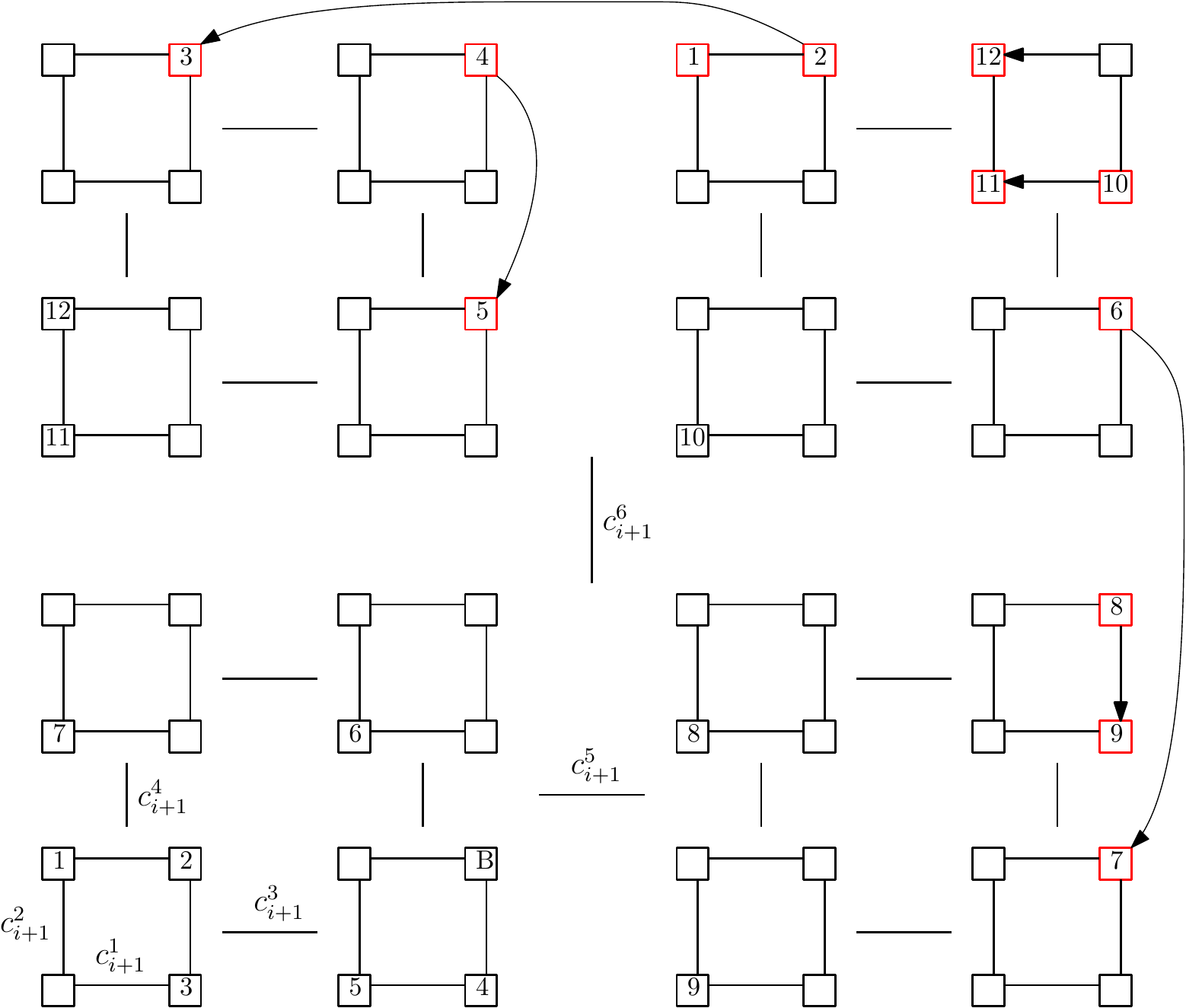}
\caption{The frame $F_3$. Again, only the backward edges are explicitly oriented.}
\label{fig:zadehF3}
\end{figure}

Following, we give three crucial properties that will hold for our construction. The first one is about the base case $A_0$ (refer to Figure~\ref{fig:zadehBC}).
\begin{enumerate}
\item[(i)] There is at least one  $(D_0)$-saturated vertex, other than the starting vertex $v_0^0$ and the sink $s_0$, in $A_0$. 
In addition, the sink $s_0$ of $A_0$ is at least two vertices away from the last vertex on the path $P_0$ that was $(D_0)$-saturated. 
%This is 
%easy to establish by looking at the figure of the base case AUSO. \AT{Figure only in APPENDIX}
\end{enumerate}
%Let us explain what the above means for us. Assume that the algorithm is solving $A_{i+1}$ and  \AT{Write this properly. Maybe write as lemma with proof in appendix?}
%let the token be at a vertex $v_s$ which is both $IN$- and $OUT$-saturated.
%So, for every $d \in IN\cup OUT$ with $b(d)>0$, $d$ is not available at $v_s$. 
%Since directions 
%from lower-index bundles have priority, the next direction to be used will be from $D_0$. By property (i)
%above, after this direction is used the vertex reached will not be a vertex that has all edges on coordinates $C_0$ 
%incoming and, hence, it will not be the sink $s_{i+1}$.
%This property will be used at steps (3) and (10) of the step-by-step analysis in the next section.
Property (i) will be utilized in the proof of Lemma~\ref{lem:nextnotinsat}.
%At this point let us define the base case $A_0$. This is the 6-AUSO presented in Figure~\ref{fig:zadehBC}.
%The starting vertex is $v_0^0 = \{c_0^2\} = \fbox{1}$.
The other two properties hold for every inductive step $A_k$ of the construction $0 \leq k \leq i+1$.
%\vspace{-0.2cm}
\begin{enumerate}
\item[(ii)] Let the token $t$ be at a vertex $v$ such that $v$ is $D_k^+$-saturated. Then the algorithm will use each direction from $D_k^+$
at most once and it will reach another $D_k^+$-saturated vertex.
\item[(iii)] When the token reaches the sink $s_k$ of $A_k$ there are exactly $4(k+1)$ negative imbalanced coordinates:
$IM_k = \{-c^3_0, -c^4_0,-c^5_0,-c^6_0, \ldots, -c^3_k, -c^4_k,-c^5_k,-c^6_k\}$. For every $d\in IM_k$ we have $b(d) = 1$ and for every other $d$ we have $b(d)=0$.
\end{enumerate}

In the next figure we give the 6-AUSO for the base case $A_0$ and argue that Properties (i), (ii) and (iii) hold for it. 
The arguments will appear in the caption of Figure~\ref{fig:zadehBC}.

\newpage

\begin{figure}[htbp] 
  \centering\includegraphics[width = \textwidth]{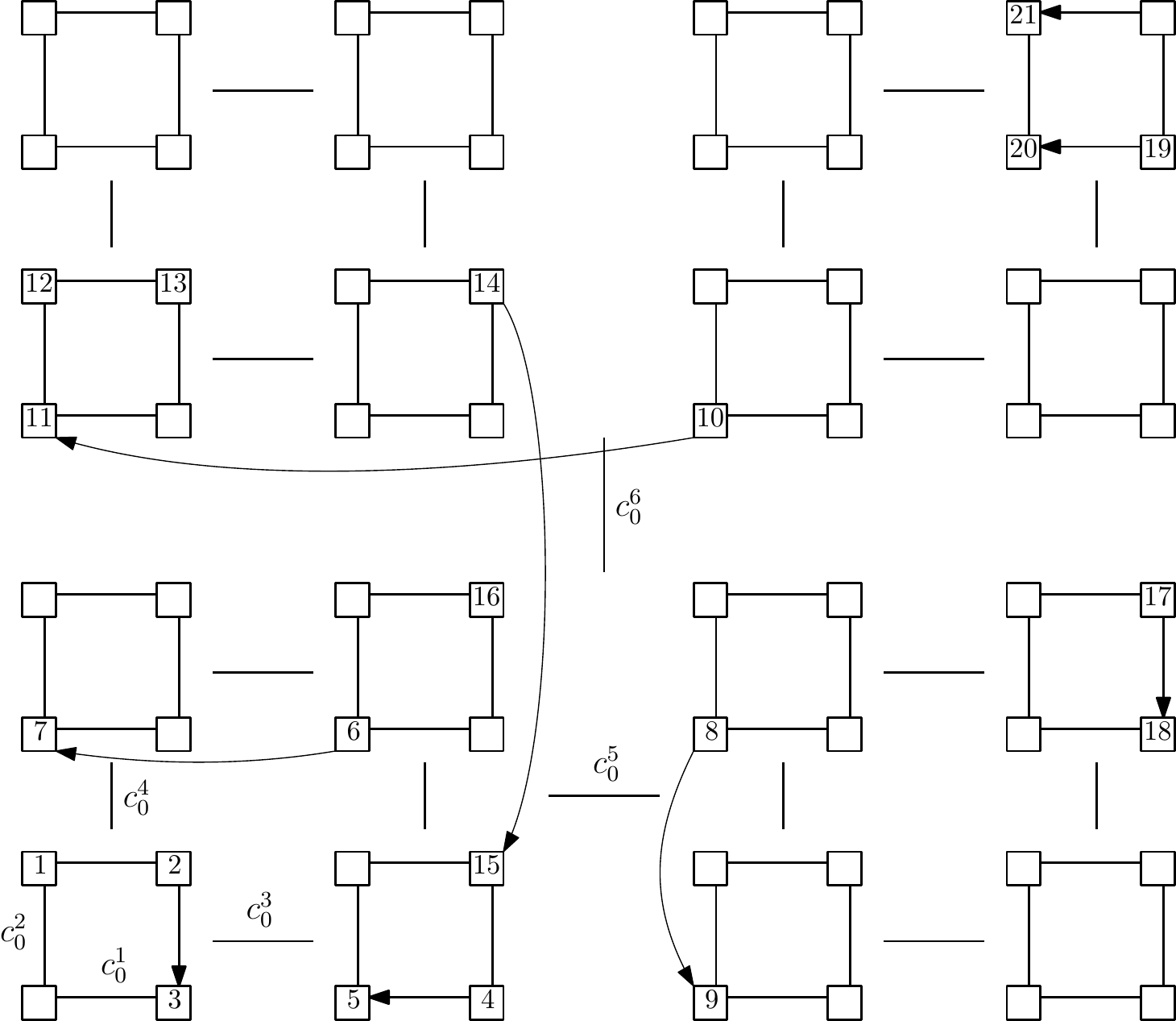}
\caption{The base case $A_0$. Only the backward edges are explicitly drawn.
The tie-breaking rule $T_0$ is essentially the same as for the inductive step. Specifically, 
$T_0 = (+1, -2, +3, -1, +4, -3, +5, -4, +6, -5, +2, -6)$. 
Here, a number $\pm k$ indicates the direction $\pm c_0^k$.
The numbers \fbox{$\cdot$} in the figure indicate the path that the token will take: 
$\fbox{1} \rightsquigarrow \fbox{21}$. Note that \fbox{12} is $D_0$-saturated while the sink is at \fbox{21} and, so, 
Property (i) is satisfied.
The initial vertex \fbox{1} is, of course, $D_0$-balanced and $D_0$-saturated. From there the algorithm uses
$|D_0|-1$ different directions to reach vertex \fbox{12} which is $D_0$-saturated. From the latter the algorithm uses 
$|D_0| -3$ different directions to reach vertex \fbox{21} which is the sink of $A_0$ and, thus, $D_0$-saturated.
Therefore, Property (ii) is satisfied for the base case. 
After the sink is reached at \fbox{21}, we have $b(-3)=b(-4)=b(-5)= b(-6)=1$ and $b(d) = 0$ for all other $d \in D_0$.
This satisfies Property (iii) for the base case.}
\label{fig:zadehBC}
\end{figure}

%Note that Properties (ii) and (iii) hold for the base case: this will be argued in the caption of Figure~\ref{fig:zadehBC}. 
In the step-by-step analysis, which comes later, we will argue that they also hold for every inductive step of the 
construction. Also note, with regards to Property (iii), that if the token takes the directions in $IM_k$ from the sink $s_k$, it will go back to the starting vertex 
$v_0^k$. Such a path is not available in any of the connecting frames; however, such a path %which spans exactly those directions 
is available in the balance-AUSO of $A_{k+1}$.
We will now analyze the behavior of Zadeh's rule on AUSO $A_{i+1}$. Firstly, let us define the tie-breaking ordered list %\AT{Stress out again that it is natural}
%$T_{i+1} = T_i \cdot (+1, -2, +3, -1, +4, -3, +5, -4, +6, -5, +2, -6)$. 
$T_{i+1}$: 
\[T_{i+1} = T_i \cdot (+1, -2, +3, -1, +4, -3, +5, -4, +6, -5, +2, -6).\]
Here, a number $\pm k$ indicates the direction $\pm c_{i+1}^k$.
Furthermore, the next lemma states that if the token is at 
a $D_{i+1}^+$-saturated vertex that does not have all edge on $C_i^+$ incoming, then there is an outgoing edge from $D_0$ to 
a neighboring vertex such that the latter is not IN-saturated. 
%consider the next lemma.

\begin{lemma} \label{lem:nextnotinsat}
Let token $t$ be at a $D_{i+1}^+$-saturated vertex $v$, such that $s(v) \cap C_i^+ \neq \emptyset$.
Then the next step of the token will be $v \xrightarrow{d} v'$ such that $d \in D_0$. Moreover, $v'$ 
is not $IN$-saturated and $s(v') \cap C_i^+ \neq \emptyset$.
\end{lemma}
\begin{proof}
Since vertex $v$ is $D_{i+1}^+$-saturated the next direction that will be chosen by the algorithm will be from the smallest index bundle 
that has an available direction at $v$. Note that it is actually not possible that
$s(v) \cap C_0 = \emptyset$, $v$ is $D_{i+1}^+$-saturated and $s(v) \cap C_i^+ \neq \emptyset$ at the same time: if $v$ was such that $s(v) \cap C_0 = \emptyset$ then, by 
Property (i), the directions from $IM_0$ would not be balanced. By design the token would go to a balance-AUSO where the directions from $|IM_0|$ are available.
We conclude that 
$s(v) \cap C_0 \neq \emptyset$. 

Then, the next direction chosen by the algorithm at $v$ will be $d \in D_0$, $v \xrightarrow{d} v'$. After this move we will have that 
$b(d') > 0$ for every $d' \in D_{i+1}^+$, $d' \neq d$. Thus, $v'$ will not be $IN$-saturated. In addition, 
by Property (i), $v'$ will be such that $s(v) \cap C_0 \neq \emptyset$. 
\end{proof}
%In the next figure we give the 6-AUSO for the base case $A_0$ and we argue that Properties (i), (ii) and (iii) hold for it. 

%\begin{figure}[htbp] 
 % \centering\includegraphics[width = \textwidth]{zadeh_BC}
%\caption{The base case $A_0$. Only the backward edges are explicitly drawn.
%The tie-breaking rule $T_0$ is essentially the same as for the inductive step. Specifically, 
%$T_0 = (+1, -2, +3, -1, +4, -3, +5, -4, +6, -5, +2, -6)$. 
%Here, a number $\pm k$ indicates the direction $\pm c_0^k$.
%The numbers \fbox{$\cdot$} in the figure indicate the path that the token will take: 
%$\fbox{1} \rightsquigarrow \fbox{21}$. Note that \fbox{12} is $D_0$-saturated while the sink is at \fbox{21} and, so, 
%Property (i) is satisfied.
%The initial vertex \fbox{1} is, of course, $D_0$-balanced and $D_0$-saturated. From there the algorithm uses
%$|D_0|-1$ different directions to reach vertex \fbox{12} which is $D_0$-saturated. From the latter the algorithm uses 
%$|D_0| -3$ different directions to reach vertex \fbox{21} which is the sink of $A_0$ and, thus, $D_0$-saturated.
%Therefore, Property (ii) is satisfied for the base case. 
%After the sink is reached at \fbox{21}, we have $b(-3)=b(-4)=b(-5)= b(-6)=1$ and $b(d) = 0$ for all other $d \in D_0$.
%This satisfies Property (iii) for the base case.}
%\label{fig:zadehBC}
%\end{figure}

Moreover, we one last notation that will help us with the upcoming step-by-step analysis:
Let us denote with $d^{OUT}_{max}$ the direction that maximizes history over the $OUT$ directions;
similarly, we define  $d^{IN}_{max}$.
%Assume that the token is at an $IN$-saturated vertex and $b(d^{OUT}_{max}) > 0$. %Then, 
%The algorithm will use directions from $OUT$ until it reaches a vertex that is $OUT$-saturated.

\subparagraph{Step-by-step analysis.} We are ready to give a description for the behavior of the algorithm on $A_{i+1}$. 
The token is initially at the starting vertex $v_0 = v_0^{i+1} = \{c_0^2, \ldots, c_{i+1}^2\}$. In the rest we write 
$\pm k$ to mean $c_{i+1}^k$. 

Note that in the analysis below we will also argue that Properties (ii) and (iii) are inductively satisfied. For the latter a very simple 
accounting of the imbalanced coordinates at the end of the run of the algorithm is enough; this will be provided at Step (10).
For the former we basically have to show that when the token is at a $D_{i+1}^+$-saturated vertex then it will go through 
at most $|D_{i+1}^+|-1$ many different directions to reach another $D_{i+1}^+$-saturated vertex. For the directions from $IN$
the inductive hypothesis will be employed. For the directions from $OUT$ we will explicitly do the accounting. 
This will be found at Steps (2), (5), (8) and (10) of the analysis.

\begin{enumerate}[(1)]
\item 
The token is at \fbox{1}. Directions from $IN$ will be used, since they have priority in $T_{i+1}$. % $b(d^{OUT}_{max}) = 0$. 
After some steps, an $IN$-saturated vertex $v_s$ will be reached, by Property (ii) which holds for $A_i$ by the inductive hypothesis.
At $v_s$, we have that $b(d^{OUT}_{max}) = 1$. The connecting frame will be $F_2$. 
The directions from $OUT$ will be utilized and used in the order defined by $T_{i+1}$.
The token will take a path $\fbox{1} \rightsquigarrow \fbox{12}$, where it will reach $v_s \bot \fbox{12}$.
%This is possible because the connecting frame for $v_s$ is $F_2$ and the dotted edges are available. 

\item When the token reaches \fbox{12} we have: 
$b(-6)=1$ and for every other direction $d\in OUT$, $b(d)=0$; also, $b(d^{OUT}_{max}) =0$.
The frame is $F_2$ and the dashed edge is not available: $v_s \bot \fbox{12} \leftarrow v_s \bot \fbox{1}$.
This means that $v_s \bot \fbox{12}$ is $D_{i+1}^+$-saturated (it is $OUT$-saturated and $IN$-saturated). 
We conclude that Property (ii) also holds for the new inductive step  up to this point 
(the algorithm has used each direction from OUT at most once since the previous $D_{i+1}^+$-saturated vertex). 

Then, the algorithm will use one direction from $IN$. The next vertex $v$ is not $IN$-saturated, by Lemma~\ref{lem:nextnotinsat}. Thus, the connecting frame 
for $v$ is $F_1$. Direction $-6$ (which is the only direction that has $b(-6)=2$ at this point) 
will be used and the token will go to \fbox{1}. Afterwards, $b(OUT,d)=0$, for every $d \in OUT$, 
$b(d^{OUT}_{max}) =1$ and $b(d^{IN}_{max}) =0$ (because some direction from $IN$ has been used once more). 

\item The algorithm will keep looping between steps (1)-(2) until the token reaches 
a vertex $v_{s_i}$ such that $s(v_{s_i}) \cap C_i^+ = \emptyset$.
That is a vertex that corresponds to the sink of $A_i$. This will be evaluated in \fbox{1}, by Lemma~\ref{lem:nextnotinsat}.
At that vertex, we have that $b(d) = 1$, for every $d \in IM_i$, by Property (iii) which holds inductively.

\item The connecting frame for $v_{s_i}$ is $F_3$. No direction from $IN$ is available and, thus, the algorithm will choose +1, first, and then +3
(only positive directions are available). The token is now at vertex $v_{s_i}\bot \fbox{B}$.
In \fbox{$B$} we have a uniform orientation with the sink at vertex $v_0 \bot \fbox{B}$.
The token will move over all the imbalanced directions from $IM_i$, towards the sink of \fbox{$B$}, in exactly $|IM_i|$ steps.
After this, $b(d)=0$ for every $d\in IN$ and the token is at $v_0 \bot \fbox{B}$. 
The connecting frame is $F_2$ (since $v_0 \bot \fbox{B}$ is $IN$-saturated because $IN$ is balanced).
For the directions from $OUT$ we have that $b(+1)=b(+3)=0$ and $b(d) = 1$ for every other $d\in OUT$.

\item  In the next step, the algorithm will choose +4 and then +5 and the token will go to vertex $v_0 \bot \rbox{8}$. From there on, backwards edges 
will become again available.
The token takes a path $\rbox{8} \rightsquigarrow \rbox{12}$, according to $T_{i+1}$, where it ends up at vertex $v_0 \bot \rbox{12}$.
%Since, $v_0 \bot \rbox{12}$ is $IN$-saturated (because $IN$ is balanced at this vertex), 
The connecting frame is $F_2$ and the dashed edge is not available: $v_0\bot \rbox{12} \leftarrow v_0\bot \rbox{1}$.
This means that $v_0\bot \rbox{12}$ is also $OUT$-saturated and, thus, $D_{i+1}^+$-saturated. More specifically, we have 
$b(-3)=b(-4)=b(-5)=b(-6)=1$ and for every other $d\in OUT$, $b(d) = 0$. The algorithm has used strictly less than
$|OUT|$ directions from $OUT$, each one different to the other, and Property (ii) is satisfied up to this point.

\item[$\blacktriangleright$] The token is currently at vertex $v = v_0 \bot \rbox{12}$, which is such that $v \cap C_i^+ = v_0^i$.
Moreover, $IN$ is balanced.
Since the algorithm is deterministic, all the steps that the token will take 
using directions from $IN$ will be consistent with the path $P_i$. This is because using directions 
from $OUT$ according to the tie-breaking rule $T_{i+1}$ it will always be in a face from the set $\{\rbox{1}, \ldots \rbox{12}\}$ and 
all these faces are oriented according to $A_i$.

\item In the next step, the algorithm will use one direction $d\in IN$, $v_0 \bot \rbox{12} \xrightarrow{d} v$.
The next vertex $v$ is not $IN$-saturated, by Lemma~\ref{lem:nextnotinsat} and, thus, the connecting frame 
is $F_1$. The algorithm will use direction $-3$ and the token will go to \rbox{1}. So, now 
$b(-4)=b(-5)=b(-6)=1$ and for every other $d\in OUT$, $b(d) = 0$; so, $b(d^{OUT}_{max}) = 0$. 

\item The token is at $\rbox{1}$. Directions from $IN$ will be used, since they have priority in
$T_{i+1}$. % and because
%no negative directions from $OUT$ are available.
From $OUT$ only negative directions are imbalanced but no negative direction is available at $\rbox{1}$.
After some steps, an $IN$-saturated vertex will be reached, by Property (ii) which holds for $A_i$ by the inductive hypothesis.
Let us call this vertex $v_s$.
When this happens, $b(d^{OUT}_{max}) =1$ and $b(OUT, -4)= b(OUT, -5) = b(OUT,-6)=1$.
The directions from $OUT$ will be utilized and used in the order defined by $T_{i+1}$. 
The token will take a path $\rbox{1} \rightsquigarrow \rbox{12}$. 
This is possible because the connecting frame for $v_s$ is $F_2$ and the dotted edges are available.
% After this, the balance of $OUT$ will be the same as in the end of step (5) above.

\item When the token reaches \rbox{12} we have $b(-3)=b(-4)=b(-5)=b(-6)=1$ and for every other $d\in OUT$, $b(d) = 0$. 
The frame is still $F_2$ and the dashed edge is not available: $v_s\bot \rbox{12} \leftarrow v_s\bot \rbox{1}$.
This means that $v_s \bot \rbox{12}$ is also $OUT$-saturated and, thus, it is $D_{i+1}^+$-saturated. 
We conclude that up to this point Property (ii) also holds for the new inductive step
(the algorithm has used each direction from OUT exactly once since the previous $D_{i+1}^+$-saturated vertex). 

Then, the algorithm will use one direction from $IN$. 
The next vertex $v$ is not $IN$-saturated, by Lemma~\ref{lem:nextnotinsat}. 
Thus, the connecting frame for $v$
is $F_1$. Direction $-3$ will be used and the token will go to \rbox{1}; $b(d^{OUT}_{max}) = 1$. 
%Afterwards, 
%$b(-4)=b(-5)=b(-6)=1$ and for every other $d\in OUT$, $b(d) = 0$; also, $b(d^{OUT}_{max}) = 0$. 

\item The algorithm will keep looping between steps (7)-(8) until the token reaches 
a vertex $v_{s_i}$ such that $s(v_{s_i}) \cap C_i^+ = \emptyset$. %(equivalently, $v_{s_i} \cap C_i^+ = s_i$).
That is a vertex that corresponds to the sink of $A_i$. It will be evaluated in \rbox{1}, by Lemma~\ref{lem:nextnotinsat}.

\item The token is on vertex $v_{s_i} \bot \rbox{1}$; hence, the connecting frame is $F_3$. 
By the same arguments as for step (7), when the token reaches $v_{s_i}$ we have 
$b(d^{OUT}_{max}) =1$ and $b(OUT, -4)= b(OUT, -5) = b(OUT,-6)=1$. The token will take a path 
$v_{s_i} \bot \rbox{1} \rightsquigarrow v_{s_i} \bot\rbox{12}$. The global sink will be at 
$v_{s_i} \bot \rbox{12}$. 
After the sink has been reached the balance for $OUT$ will be the same as the beginning of step (8) above. Namely, 
$b(-3)=b(-4)=b(-5)=b(-6)=1$ and for every other $d\in OUT$, $b(d) = 0$.  Thus, Property (iii) will be satisfied inductively.
Finally, the global sink is $D_{i+1}^+$-saturated and the algorithm has used 
each direction from $OUT$ exactly once since the previous $D_{i+1}^+$-saturated vertex.
%(the algorithm has used each direction from OUT exactly once since the previous $D_{i+1}^+$-saturated vertex). 
%$|OUT|-1$ different coordinates on the path 
%$v_{s_i} \bot \rbox{1} \rightsquigarrow v_{s_i} \bot\rbox{12}$. 
Therefore, Property (ii) is satisfied. 
\end{enumerate}
With the above analysis, we have proved that the path $P_{i+1}$ will have length that is larger than twice the length of path $P_i$.
Therefore, we obtain the recursion $T(n+6) > 2T(n)$ which leads to the proof of Theorem~\ref{thm:zadeh}.

\section{Conclusions}
In this paper, we have constructed AUSO on which the three pivot rules of interest will take exponentially long paths. 
Several interesting problems remain open: %First of all, it is not clear if our constructions can be realized as LPs. 
First and foremost is settling if Zadeh's and Johnson's rules admit exponential lower bounds even on linear programs. 
%Of course, non-trivial upper bounds for this rule would also be of great interest. 
Moreover, it remains open to decide if Zadeh's rule admits 
Hamiltonian paths on AUSO, a direction suggested by Aoshima \etal \cite{aoshima}.
Finally, we are interested in exponential lower bounds for all the history-based rules that are discussed in \cite{aoshima}. 
We believe that our methods can be used to prove exponential lower bounds on AUSO for all of these rules. 

\bibliography{uso}

\end{document}